\documentclass[]{fundam}
\usepackage{amsmath,amssymb}
\usepackage{makeidx,epsfig,lscape}
\usepackage{color,colortbl}
\usepackage{fancyhdr}
\usepackage{xcolor,pict2e}

\usepackage[T1]{fontenc}
\usepackage[latin1]{inputenc}
\usepackage{pifont}
\usepackage{pslatex}
\usepackage[mathscr]{euscript}
\usepackage{latexsym}
\usepackage{stmaryrd}
\usepackage{amssymb}
\usepackage{amsmath}
\usepackage{graphicx}
\usepackage{pst-all}
\usepackage{verbatim}
\usepackage{fancyvrb}
\usepackage{mathrsfs}
\usepackage{ifsym}
\usepackage{algorithm}
\usepackage{algpseudocode}
\usepackage{mathtools}
\usepackage{listings}
\usepackage[framemethod=TikZ]{mdframed}
\usepackage{enumitem}
\usepackage{textcomp}
\usepackage{csquotes}
\usepackage{url}

\let\oldComment=\Comment
\renewcommand{\Comment}[1]{\oldComment{{\scriptsize#1}}}
\algdef{SE}[DOWHILE]{DoWhile}{EndDoWhile}{\algorithmicdo}[1]{\algorithmicwhile\ #1}%

\mdfdefinestyle{MyFrame}{%
	linecolor=white,
	outerlinewidth=0pt,
	roundcorner=0pt,
	innertopmargin=4pt,
	innerbottommargin=4pt,
	innerrightmargin=4pt,
	innerleftmargin=4pt,
	leftmargin=4pt,
	rightmargin=4pt
}

\begin{document}

%
\setcounter{page}{1}
\publyear{2021}
\papernumber{0001}
\volume{178}
\issue{1}
%

\title{A Projection-Stable Grammatical Model for the Distributed Execution of Administrative Processes with Emphasis on Actors' Views}

\address{ndadjimaxime@yahoo.fr (Zekeng Ndadji Milliam Maxime)}

\author{Milliam Maxime \textsc{Zekeng Ndadji} $^{a,b}$ \and Maurice \textsc{Tchoup\'e Tchendji} $^{a,b}$ \and Cl\'ementin \textsc{Tayou Djamegni} $^{a}$ \and Didier \textsc{Parigot} $^{c}$\\
	$^{a}$ Department of Mathematics and Computer Science, 
	University of Dschang\\
	$^{b}$ FUCHSIA Research Associated Team, \url{https://project.inria.fr/fuchsia/}\\
	$^{c}$ Inria, Sophia Antipolis, France\\
	\{ndadji.maxime, maurice.tchoupe\}@univ-dschang.org, dtayou@yahoo.com, didier.parigot@inria.fr 
} 

\maketitle

\runninghead{M. Zekeng et al.}{A Grammatical Model for the Distributed Execution of Administrative Processes}

\begin{abstract}
  During the last two decades, the decentralized execution of business processes has been one of the main research topics in Business Process Management. Several models (languages) for processes' specification in order to facilitate their distributed execution, have been proposed. LSAWfP is among the most recent in this area: it helps to specify administrative processes with grammatical models indicating, in addition to their fundamental elements, the permissions (reading, writing and execution) of each actor in relation to each of their tasks. 
  In this paper, we present a model for a completely decentralized and artifact-centric execution of administrative processes specified using LSAWfP. The presented model puts particular emphasis on actors' views: it then allows the confidential execution of certain tasks by ensuring that, each actor potentially has only a partial perception of the processes' global execution states. The model thus solves a very important problem in business process execution, which is often sidelined in existing approaches. To accomplish this, the model rely on three projection algorithms allowing to partially replicate the processes' global execution states at a given moment, to consistently update the obtained partial states and to deduce new coherent global states. The proposal of these three algorithms, the proof of underlying mathematical tools' stability and a proposal of their implementation, are this paper's main contributions.
\end{abstract}


\section{Introduction}
\label{Introduction}

Administrative business processes are those of which all cases are known and predictable; that is to say, tasks sequencing rules are simple and clearly defined \cite{mcCready, van1998application}: they are the most frequently encountered in practice \cite{van2013business, dumas2018fundamental} (the peer review process \cite{van2001proclets, badouel14, zekeng2020alanguage, zekeng2020lsawfp, ndadji2020language}, the insurance claims process \cite{esparza2016reduction, cartledge2020system}, etc.). In its most widespread approach, Business Process Management (BPM) technology breaks down the automation of a given business process to its formal specification (modelling) in a \textit{workflow language} \cite{dumas2018fundamental}. The resulting process (workflow) model generally describes all the process's tasks, the control flow that link them (routing) and the actors in charge of executing them \cite{van2013business}. 

For the decentralized execution of an administrative process described in any workflow language, one can imagine a distributed \textit{Workflow Management System} (WfMS) made up of several reactive agents or peers, driven by human agents (actors in charge of executing tasks), coordinating with each other using an artifact that they cooperatively edit. In fact, the distributed execution of an administrative process is similar to the cooperative edition of a form that has to circulate from site to site (mobile form) in order to be edited by the different actors of the process. Upon its arrival on a site, the actor associated with the site must be able to examine it and deduce without ambiguity, the already edited fields (these correspond to the already executed tasks of the process), the fields that he must immediately edit (these correspond to the ready tasks that he must execute), and possibly, the sites to which he must return/redirect the form for further processing at the end of its edition.  It is easy to imagine that there could be forms with independent fields that need to be filled in by different actors. In this case, in order to speed up processing, it is acceptable that at a given time, there may be several replicas of the form that are simultaneously edited in the system.

A major characteristic of administrative processes is confidentiality: not all actors in an administrative process are necessarily aware of all processing and/or data generated in the process. It is therefore natural to assume that the form that is presented to an actor for editing on a site is only a potentially partial replica of the global form\footnote{The global form contains all the data already filled in so far by the various actors in the system. It therefore gives the process' global execution status at a given moment, by explicitly highlighting the fields already edited, those ready to be immediately edited, and those that will be edited later (in the case of dependencies between fields).}; this (partial) replica only contains information (relating to processing and data) that is of proven interest to the considered actor. Once a partial replica of a form has been received, it is essential to ensure that all editing actions on it can be consistently integrated into the global form. In order to satisfy this constraint, it is sufficient for each actor to have a local "supervisor" who must control his editing actions. The confidential execution of certain tasks as well as the access restriction to certain process data are generally only very rarely addressed in existing process management approaches. Most of the studies that do mention these aspects, either relegate them to second place or do not treat them formally like other aspects. In this study, the focus is on these concerns.

The form described in the two paragraphs above can be seen as a structured document (a tree) circulating from site to site, to be extended by cooperative editions made at the level of its leaves (positive edition\footnote{In a positive edition, no information is erased \cite{saito2005optimistic, shapiro2018Optimistic}. Editing actions on the document have the effect of making the tree representing it grow, by adding sub-trees at the level of its leaves.}). The nodes of this tree therefore represent either the tasks already executed, or those ready to be executed; moreover, the relations between nodes (father-son, brother-brother) correspond exactly to the ordering of tasks. This tree is called "\textit{artifact}" in the artifact-centric approach to business process modelling.

Based on the cooperative editing model of structured documents studied in the works of Badouel et al. \cite{badouelTchoupeCmcs, theseTchoupe, tchoupeAtemkeng2, tchoupezekeng2016reconciliation, tchoupeZekeng2017, zekengTchoupe2018}, we propose in this paper, a model for the distributed execution of administrative processes (cooperative edition of artifacts) that relies on algorithms to obtain : (1) partial replicas of the global artifact (\textit{artifact projection} algorithm); these contain only the information relevant to the considered actors; (2) local models that constrain local editing actions on local artifacts (\textit{model projection} algorithm), so as to ensure that these are always "expandable" as editing actions (updates) on the global artifact (\textit{expansion} algorithm). The proposed execution model applies only to process models obtained using the language LSAWfP (\textit{A Language for the Specification of Administrative Workflow Processes}) \cite{zekeng2020alanguage, zekeng2020lsawfp, ndadji2020language}. LSAWfP is a new language designed for the specification (using a variant of attributed grammars) of administrative processes with particular emphasis placed on the modelling (using views) of the perceptions that the various actors have on processes and their data. The specification of a particular process in this language is given by a triplet (\textit{a Grammatical Model of Administrative Workflow Process} - GMAWfP -) $\mathbb{W}_f = \left(\mathbb{G}, \mathcal{L}_{P_k}, \mathcal{L}_{\mathcal{A}_k} \right)$ wherein, $\mathbb{G}$ is the model of tasks and their sequencing, $\mathcal{L}_{P_k}$ and $\mathcal{L}_{\mathcal{A}_k}$ represent respectively the list of actors and their accreditations\footnote{The accreditation of a given actor provides information on its rights (permissions) relatively to each sort (task) of the studied process.}. If we consider a GMAWfP $\mathbb{W}_f = \left(\mathbb{G}, \mathcal{L}_{P_k}, \mathcal{L}_{\mathcal{A}_k} \right)$ to be executed in a decentralized manner, then the main contributions of this paper are as follows:

\noindent 1) The proposal of algorithms for:
\begin{itemize}
	\item \textit{artifacts' projection}; which, given an artifact $t$ that conforms to the (global) model $\mathbb{G}$ and the accreditation in reading $\mathcal{A}_{A_i}(r)$ (known as \textit{view} and denoted $\mathcal{V}_i$) of an actor $A_i$, allows to find its partial replica $t_{\mathcal{V}_i}$; 
	
	\item \textit{model's projection}; which permits, from a view $\mathcal{V}_i$ and a global model $\mathbb{G}$, to derive a local model $\mathbb{G}_{\mathcal{V}_i}$. $\mathbb{G}_{\mathcal{V}_i}$ will guide the actions carried out by a given actor on partial replicas from his site, in order to ensure consistency with respect to the global model $\mathbb{G}$.
	
	\item \textit{expansion-pruning}; which enable the inverse projection of a partial replica $t_{\mathcal{V}_i}^{maj}$ updated by a given actor $A_i$ according to a local GMWf $\mathbb{G}_{\mathcal{V}_i}$; the goal is to integrate the contributions made by the local actor into an artifact $t_{f}$ that conforms to the global model $\mathbb{G}$.
\end{itemize}
\noindent 2) A study of \textit{stability} properties of artifacts and their model, when using the proposed algorithms.

\noindent 3) A Haskell implementation of the proposed algorithms.

~

\noindent\textbf{\textit{Organization of the manuscript}}: in the remainder of this manuscript, we briefly present the LSAWfP language and an example of a process modelled using it (sec. \ref{sec:LSAWfP}). We then present the artifact-centric model of process execution considered in this paper, in order to motivate the need to propose stable projection algorithms (sec. \ref{sec:executionModel}). We continue by proposing versions of the three projection algorithms covered in this paper as well as a study of some of their properties (sec. \ref{sec:projectionAlgorithms}). We end with a discussion and a conclusion.

\section{On the Modelling and the Execution of Administrative Business Processes using LSAWfP}
\label{sec:Préliminaire}
Several tools have been developed to address process modelling. Among the most well-known are the BPMN standard (\textit{Business Process Model and Notation}) \cite{BPMN} which uses a formalism derived from that of statecharts, and the YAWL language (\textit{Yet Another Workflow Language}) \cite{van2005yawl, van2013business, van2015business} based on Petri nets. Despite the significant research progress around these workflow languages, they are often criticized for having a much too great expressiveness compared to the needs of professionals in the field \cite{zur2013much}, for not being based on solid mathematical foundations and/or for not being intuitive \cite{borger2012approaches}: This justifies the emergence of several other languages such as LSAWfP \cite{zekeng2020alanguage, zekeng2020lsawfp, ndadji2020language}. 
In this section, we present the LSAWfP language and its illustration: the specification of the administrative process used as a running example along this paper. The artifact-centric execution model of LSAWfP's specifications is also presented in order to motivate the current work.

\subsection{A Running Example: the Peer-Review Process}
\label{sec:runningExample}
The peer-review process is a commonly used example of business process to illustrate workflow languages \cite{van2001proclets, badouel14, zekeng2020alanguage, zekeng2020lsawfp, ndadji2020language}. We choose it in this manuscript because it is easy to describe and (that's the most important) it perfectly illustrates the concepts that we handle. The description that we consider is the same as the one in \cite{zekeng2020alanguage}: the process involves four actors (an editor in chief - $EC$ -, an associated editor - $AE$ - and two referees - $R1$ and $R2$ -) coordinating to evaluate a manuscript (paper) submitted for their review.
\begin{itemize}
	\item The process starts when the $EC$ receives the paper; 
	\item Upon receipt, the $EC$ pre-validates the paper (let us call this \textbf{task "$A$"}); after the pre-validation, he can accept or reject the paper for various reasons (subject of minor interest, paper not within the journal scope, non-compliant format, etc.);
	\item If he rejects the paper, he writes a report (\textbf{task "$B$"}) then notifies the corresponding author (\textbf{task "$D$"}) and the process ends;
	\item Otherwise, he chooses an $AE$ and sends him the paper;
	\item The $AE$ prepares the manuscript (\textbf{task "$C$"}) forms a referees committee (two members in our case) and then triggers the peer-review process (\textbf{task "$E$"});
	\item Each referee reads, seriously evaluates the paper (\textbf{tasks "$G1$"} and \textbf{"$G2$"}) and sends back a report (\textbf{tasks "$H1$"} and \textbf{"$H2$"}) and a message (\textbf{tasks "$I1$"} and \textbf{"$I2$"}) to the $AE$;
	\item After receiving reports from all referees, the $AE$ takes a decision and informs the $EC$ (\textbf{task "$F$"}) who sends the final decision to the corresponding author (\textbf{task "$D$"}).
\end{itemize}

\subsection{Process Modelling with LSAWfP}
\label{sec:LSAWfP}
LSAWfP is a recent language specifically designed for administrative process modelling. It relies on a variant of attributed grammars to provide a framework for the modelling of the main conceptual aspects of such processes: these are the aspects related to tasks' scheduling (the \textit{lifecycle model}), to data consumed and produced by tasks (the \textit{informational model}), and to resources in charge of executing tasks (the \textit{organizational model}) \cite{divitini2001inter}. In addition, LSAWfP puts a particular emphasis on the modelling (using views) of the perceptions that the various stakeholders have on processes and their data in order to guarantee confidentiality. 
To model a given process with LSAWfP, four key steps must be followed: (1) model the process scenarios using annotated trees and (2) derive from annotated trees, an abstract grammar which will be used as lifecycle model; then (3) identify the actors involved in the process execution and (4) establish the role played by each of them using a list of accreditations.

\subsubsection{Modelling Process Scenarios using Artifacts}
\label{sec:artifacts}
LSAWfP is founded on the principle that, by definition, all execution scenarios, all actors and the role they play in relation to tasks of a given administrative process $\mathcal{P}_{ad}$, are known in advance. LSAWfP therefore proposes to model each $\mathcal{P}_{ad}$'s execution scenario using an annotated tree $t_i$ called \textit{target artifact}. In such a tree, each node (labelled $X_i$) potentially corresponds to a task $X_i$ of $\mathcal{P}_{ad}$ and each hierarchical decomposition (a node and its sons) corresponds to a scheduling: the task associated with the parent node must be executed before those associated with the son nodes; the latter must be executed according to an order - parallel or sequential - that can be specified by particular annotations "$\fatsemi$" (is sequential to) and "$\parallel$" (is parallel to) which will be applied to each hierarchical decomposition. The annotation "$\fatsemi$" (resp. "$\parallel$") reflects the fact that the tasks associated with the son nodes of the decomposition must (resp. can) be executed in sequence (resp. in parallel).

For the running example (the peer-review process), there is only two execution scenarios: the one in which the $EC$ immediately rejects the paper and the one in which the paper goes through the validation process. These can be modelled using the two artifacts $art_1$ and $art_2$ in figure \ref{fig:artefactsGlobaux}. In particular, we can see that $art_1$ shows how the task "Receipt and pre-validation of a submitted paper" assigned to the $EC$, and associated with the symbol $A$ (see sec. \ref{sec:runningExample}), must be executed before tasks associated with the symbols $B$ and $D$ that are to be executed in sequence.

\begin{figure}[ht!]
	\centering
	\includegraphics[scale=0.36]{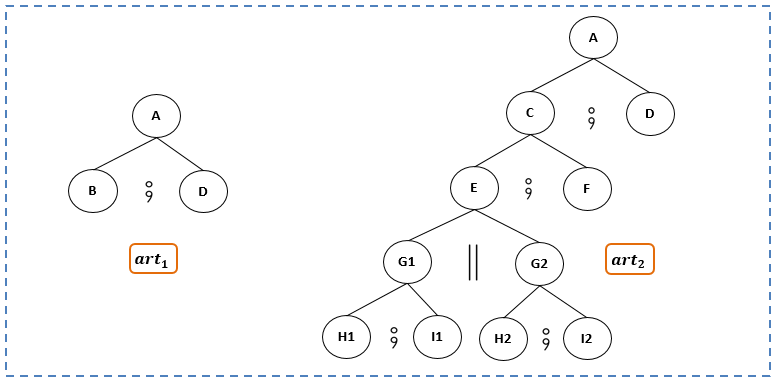}
	\caption{Example of target artifacts for a given process (peer-review process)}
	\label{fig:artefactsGlobaux}
\end{figure}

\subsubsection{Deducing the Grammatical Model of Workflow (GMWf)}
\label{sec:GMWf}
After modelling the scenarios of the studied process using target artifacts, LSAWfP suggests to extract from them, an abstract grammar called a \textit{Grammatical Model of Workflow (GMWf)}. This is done by simply substituting the set of target artifacts by a grammar $\mathbb{G}$ (a GMWf) in which, each symbol refers to a task and, each production $p$ is of one of the following two forms: $p: X_0 \rightarrow X_1 \fatsemi \ldots \fatsemi X_n$ or  $p: X_0 \rightarrow X_1 \parallel \ldots \parallel X_n$. These two forms of productions perfectly models the two types of ordering (parallel or sequential) retained in the design of the target artifacts. In this case, each target artifact $t_i$ \textit{is conform to} $\mathbb{G}$ and we note $t_i \therefore \mathbb{G}$; also the root symbols of the different target artifacts make up the set of axioms of $\mathbb{G}$. A GMWf can then be formally defined as follows:
\begin{definition}
	\label{defGMWf1}
	A \textbf{Grammatical Model of Workflow} (GMWf) is defined by $\mathbb{G}=\left(\mathcal{S},\mathcal{P},\mathcal{A}\right)$
	where :
	\begin{itemize}
		\item $\mathcal{S}$ is a finite set of \textbf{grammatical symbols} or \textbf{sorts} corresponding to various \textbf{tasks} to be executed in the studied business process; 
		\item $\mathcal{A}\subseteq \mathcal{S}$ is a finite set of particular symbols called \textbf{axioms}, representing tasks that can start an execution scenario (roots of target artifacts), and 
		\item $\mathcal{P}\subseteq\mathcal{S}\times\mathcal{S}^{*}$ is a finite set of \textbf{productions} decorated by the annotations "$\fatsemi$" (is sequential to) and "$\parallel$" (is parallel to): they are \textbf{precedence rules}. 
		A production $P=\left(X_{P(0)},X_{P(1)},\cdots X_{P(|P|)}\right)$ is either of the form $P: X_0 \rightarrow X_1 \fatsemi \ldots \fatsemi X_{|P|}$, or of the form $P: X_0 \rightarrow X_1 \parallel \ldots \parallel X_{|P|}$ and $\left|P\right|$ 
		designates the length of $P$ right-hand side.
		A production with the symbol $X$ as left-hand side is called a \textit{X-production}.
	\end{itemize}
\end{definition}

Considering the case of the peer-review process whose target artifacts are represented in figure \ref{fig:artefactsGlobaux}, the derived GMWf is  $\mathbb{G}=\left(\mathcal{S},\mathcal{P},\mathcal{A}\right)$ in which the set $\mathcal{S}$ of grammatical symbols is
$\mathcal{S}=\{A, B, C, D, E, F,$ $G1, G2, H1, H2, I1, I2\}$ (see sec. \ref{sec:runningExample});
the only initial task (axiom) is $A$ (then $\mathcal{A}=\{A\}$) and the set $\mathcal{P}$ of productions is:
\[ 
\begin{array}{l|l|l|l}
P_{1}:\; A\rightarrow B\fatsemi D & \; P_{2}:\; A\rightarrow C\fatsemi D\; & \; P_{3}:\; C\rightarrow E\fatsemi F\; & \; P_{4}:\; E\rightarrow G1\parallel G2    \\
P_{5}:\; G1\rightarrow H1 \fatsemi I1 & \; P_{6}:\; G2\rightarrow H2 \fatsemi I2\; & \; P_{7}:\; B\rightarrow \varepsilon\; & \; P_{8}:\; D\rightarrow \varepsilon  \\
P_{9}:\; F\rightarrow \varepsilon & \; P_{10}:\; H1\rightarrow \varepsilon & \; P_{11}:\; I1\rightarrow \varepsilon\; & \; P_{12}:\; H2\rightarrow \varepsilon  \\
P_{13}:\; I2\rightarrow \varepsilon &  &  &   \\
\end{array}
\]

\subsubsection{Identifying the Actors of the Process}
\label{sec:actors}
The identification of the actors taking part in the execution of the process is easily done with the help of its textual description. For example, according to the description of the peer-review process, four ($k=4$) actors participate in its execution: an editor in chief ($EC$), an associated editor ($AE$) and two referees ($R1$ and $R2$). So we deduce that the list of actors is $\mathcal{L}_{P_k}=\left\{EC, AE, R1, R2\right\}$.

\subsubsection{Establishing the List of Accreditations}
\label{sec:accreditations}
LSAWfP proposes a mechanism called \textit{accreditation}, inspired by the nomenclature of rights used in Unix-like systems, to ensure better coordination between actors and to eventually guarantee the confidentiality of certain actions and data. The accreditation of a given actor provides information on its rights (permissions) relatively to each sort (task) of the studied process's GMWf. There is three types of accreditation:

\noindent\textbf{1.}$~$ \textit{The accreditation in reading \textit{(r)}}: an actor accredited in reading on sort $X$ must have free access to its execution state (data generated during its execution). The set of sorts on which an actor is accredited in reading is called his \textbf{\textit{view}}. Any sort $X$ belonging to a given view $\mathcal{V}_i$ ($X \in \mathcal{V}_i$) is said to be \textit{visible}, and those not belonging to it are said to be \textit{invisible}.

\noindent\textbf{2.}$~$ \textit{The accreditation in writing \textit{(w)}}: an actor accredited in writing on sort $X$ can execute the associated task\footnote{Let's recall that the execution of a task is assimilated to the edition (extension) of a particular node in an artifact.}. Any actor accredited in writing on a sort is accredited in reading on it. 

\noindent\textbf{3.}$~$ \textit{The accreditation in execution \textit{(x)}}: an actor accredited in execution on sort $X$ is allowed to ask the actor who is accredited in writing in it, to execute it. 
More formally, an accreditation is defined as follows:
\begin{definition} 
	\label{defSyllabaire}
	An \textbf{accreditation} $\mathcal{A}_{A_i}$ defined on the set $\mathcal{S}$ of grammatical symbols for an actor $A_i$, is a triplet $\mathcal{A}_{A_i}=\left(\mathcal{A}_{A_i(r)},\mathcal{A}_{A_i(w)},\mathcal{A}_{A_i(x)}\right)$ such that, 
	$\mathcal{A}_{A_i(r)} \subseteq \mathcal{S}$ also called \textbf{view} of actor $A_i$, is the set of symbols on which $A_i$ is accredited in reading, 
	$\mathcal{A}_{A_i(w)} \subseteq \mathcal{A}_{A_i(r)}$ is the set of symbols on which $A_i$ is accredited in writing and  
	$\mathcal{A}_{A_i(x)} \subseteq \mathcal{S}$ is the set of symbols on which $A_i$ is accredited in execution.
\end{definition}

From the task assignment for the peer-review process in the running example, it follows that the accreditation in writing of the $EC$ is $\mathcal{A}_{EC(w)}=\{A, B, D\}$. 
Moreover, since he can only execute task $D$ if task $C$ (executed by the $AE$) is already executed (see artifacts $art_1$ and $art_2$, fig. \ref{fig:artefactsGlobaux}), he must be accredited in execution on $C$ to be able to request its execution; therefore, we have $\mathcal{A}_{EC(x)}=\{C\}$.
In addition, in order to be able to access all the information on the progress of the peer-review evaluation (task $C$) and synthesize the right decision to be returned, the $EC$ must be able to consult reports (tasks $I1$ and $I2$) and messages (tasks $H1$ and $H2$) of the referees, as well as the decision made by the $AE$ (task $F$). These tasks, in addition to $\mathcal{A}_{EC(w)}$\footnote{Remember: any actor accredited in writing on a sort is accredited in reading on it.} constitute the set $\mathcal{A}_{EC(r)}=\mathcal{V}_{EC}=\{A, B, C, D, H1, H2, I1, I2, F\}$ of tasks on which he is accredited in reading. Doing so for each of the other actors leads to the deductions of the accreditations represented in the table \ref{tableau:vuesActeurs} and we have $\mathcal{L}_{\mathcal{A}_k}=\left\{\mathcal{A}_{EC}, \mathcal{A}_{AE}, \mathcal{A}_{R1}, \mathcal{A}_{R2}\right\}$.
\begin{table}[ht]
	\centering
	\caption{Accreditations of the different actors taking part in the peer-review process}
	\begin{tabular}[t]{|p{4.2cm}|p{9.8cm}|}
		\hline
		\textbf{Actor} & \textbf{Accreditation} \\
		\hline
		Editor in Chief ($EC$) & $\mathcal{A}_{EC}=\left(\{A, B, C, D, H1, H2, I1, I2, F\}, \{A, B, D\}, \{C\}\right)$ \\
		\hline
		Associated Editor ($AE$) & $\mathcal{A}_{AE}=\left(\{A, C, E, F, H1, H2, I1, I2\}, \{C, E, F\}, \{G1, G2\}\right)$ \\
		\hline
		First referee ($R1$) & $\mathcal{A}_{R1}=\left(\{C, G1, H1, I1\}, \{G1, H1, I1\}, \emptyset\right)$ \\
		\hline
		Second referee ($R2$) & $\mathcal{A}_{R2}=\left(\{C, G2, H2, I2\}, \{G2, H2, I2\}, \emptyset\right)$ \\
		\hline
	\end{tabular}
	\label{tableau:vuesActeurs}
\end{table}

\noindent\textbf{Finally the Grammatical Model of Administrative Workflow Process}: 
with LSAWfP, the modelling of a process results in a triplet $\mathbb{W}_f = \left(\mathbb{G}, \mathcal{L}_{P_k}, \mathcal{L}_{\mathcal{A}_k} \right)$ (called \textit{a Grammatical Model of Administrative Workflow Process} - GMAWfP -) wherein, $\mathbb{G}$ is the GMWf, $\mathcal{L}_{P_k}$ is the list of actors and $\mathcal{L}_{\mathcal{A}_k}$ is the list of their accreditations.

\subsection{An Artifact-Centric Model for the Distributed Execution of GMAWfP}
\label{sec:executionModel}
The artifact-centric paradigm, emerged in the early 2000s, has become the most exploited current of thought for process modelling and execution (workflow management) over the last two decades \cite{van2013business}. Several works \cite{nigam2003business, abi2016towards, deutsch2014automatic, hull2009facilitating, lohmann2010artifact, assaf2017continuous, assaf2018generating, boaz2013bizartifact, badouel2015active} have been undertaken to develop this paradigm. According to it, workflow management focuses on both automated processes and data manipulated using the concept of \textit{business artifact}. A business artifact is considered as a document that conveys all the information concerning a particular case of execution of a given business process, from its inception in the system to its termination. In this section, we present an artifact-centric model for the distributed execution of GMAWfP inspired by the work of Badouel et al. on cooperative editing of structured documents \cite{badouelTchoupeCmcs, theseTchoupe, tchoupeAtemkeng2, tchoupeZekeng2017, zekengTchoupe2018}. We begin by presenting individually, the key concepts of the execution model, before examining the overall behaviour of the distributed system.

\subsubsection{Key Elements and Constraints of the Execution Model}
\textbf{\textit{The Execution Environment}}: to execute a given GMAWfP in a decentralized mode, we consider a completely decentralized (P2P) WfMS model (which we call \textit{P2P-WfMS-View}) whose instances (the peers) are installed on the sites of the various actors involved in processes execution. During the process execution, these peers communicate (sending/receiving requests/responses) by exchanging copies of a  (global) artifact said to be under execution. Such an artifact provides information on already executed tasks and on those ready to be executed.

As in the work of Badouel et al. \cite{badouelTchoupeCmcs, theseTchoupe, tchoupeAtemkeng2, tchoupeZekeng2017, zekengTchoupe2018}, we represent an artifact under execution by a tree (a structured document) that contains \textit{buds}. These indicate at a moment, the only places where contributions are expected. A bud can be either \textit{unlocked} or \textit{locked} depending on whether the corresponding task (node) is ready to be executed (edited) or not. Buds are typed; a \textit{bud of type $X \in \mathcal{S}$} is a leaf node labelled either by $X_{\overline{\omega}}$ or by $X_\omega$ depending on its state (\emph{locked} or \emph{unlocked}) (see fig. \ref{figDocBourgeons}). The local actions of a given actor will therefore have the effect of extending (editing) its received copy of the (global) artifact by developing, through appropriate productions, the different unlocked buds it contains.
\begin{figure}[ht!]
	\begin{center}
		\includegraphics[scale=0.24]{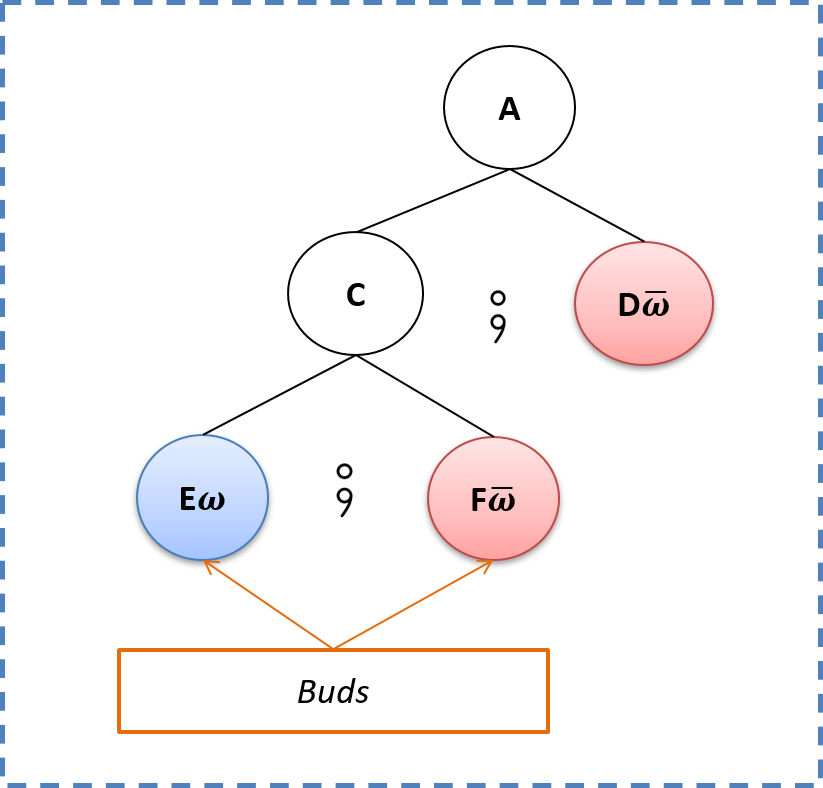}
		\vspace*{-0.4cm}
		\caption{An intentional representation of an annotated artifact containing buds.}
		\label{figDocBourgeons}
	\end{center}
\end{figure}

\textbf{\textit{The Confidential Execution of Certain Tasks}}: for confidentiality reasons, each actor acts on a potentially partial replica of the local copy of the (global) artifact: this partial replica contains only the information to which the concerned actor can have access. Technically, a partial replica $t_{\mathcal{V}_i}$ of an artifact $t$ is obtained by projection (using an operator $\pi$ said of \textbf{\textit{artifact projection}}) of $t$ according to the view $\mathcal{V}_i$ of the concerned actor: we note $t_{\mathcal{V}_i}=\pi_{\mathcal{V}_i}\left(t\right)$. 

\textbf{\textit{The Need of a Local GMWf at each Site}}: since the local actions of a particular actor depend on his perception of the process, it is necessary to control them in order not only to preserve the possible confidentiality of certain tasks, but also to ensure the consistency of local updates. To do this, one must derive a local GMWf on each site, by projecting the global GMWf according to the view of the local actor (\textbf{\textit{GMWf projection}}). This projection is carried out using $\Pi$ operator and the GMWf obtained is noted $\mathbb{G}_{\mathcal{V}_i}=\Pi_{\mathcal{V}_i}\left(\mathbb{G}\right)$. 

\textbf{\textit{The Expansion Operation}}: still with the aim of ensuring system convergence, the contributions made by a given actor and contained in an updated partial replica $t_{\mathcal{V}_i}^{maj}$, must be integrated into the local copy of the (global) artifact before any synchronization between peers. It is therefore necessary to be able to merge these two artifacts, which are based on two different models. We find here, a version of the \textbf{\textit{expansion}} problem as formulated in \cite{badouelTchoupeCmcs}.

\subsubsection{Execution Model and Peer Activity}
\label{sec:peerActivity}
Globally then, before the execution of a given process, peers are configured using its GMAWfP $\left(\mathbb{W}_f=\left(\mathbb{G}, \mathcal{L}_{P_k}, \mathcal{L}_{\mathcal{A}_k} \right)\right)$. From the global GMWf $\mathbb{G}$ and the view $\mathcal{V}_i$ of the local actor, each peer derives by projection, a local GMWf $\mathbb{G}_{\mathcal{V}_i}=\Pi_{\mathcal{V}_i}\left(\mathbb{G}\right)$.
Then, the execution of a process case is triggered when an artifact $t$ is introduced into the system (on the appropriate peer); this artifact is actually an unlocked bud of the type of one axiom $A_\mathbb{G} \in \mathcal{A}$ (initial task) of the (global) GMWf $\mathbb{G}$ (see fig. \ref{fig:systemOverview}). 
During execution, peers synchronize themselves by exchanging their local copies of the artifact being executed.

After receiving an artifact $t \therefore \mathbb{G}$ on a given peer, the latter projects it (see Peer $i$ in fig. \ref{fig:systemOverview}) according to the local view $\mathcal{V}_i$. The obtained partial replica $t_{\mathcal{V}_i} \therefore \mathbb{G}_{\mathcal{V}_i}$  is then completed (edited) when needed: the result of this edition is an artifact $t_{\mathcal{V}_i}^{maj} \therefore \mathbb{G}_{\mathcal{V}_i}$ such as $t_{\mathcal{V}_i}^{maj}$ is an update of $t_{\mathcal{V}_i}$ ($t_{\mathcal{V}_i}^{maj} \geq t_{\mathcal{V}_i}$).

At the end of the completion, the expansion-pruning of the obtained updated partial replica $t_{\mathcal{V}_i}^{maj} \therefore \mathbb{G}_{\mathcal{V}_i}$ is made in order to obtain the updated configuration $t_f \therefore \mathbb{G}$ of the (global) artifact local copy (see Peer $i$ in fig. \ref{fig:systemOverview}). If the resulting configuration shows that the process should be continued at other sites\footnote{This is the case when the artifact contains buds created on the current peer and whose actors accredited in writing are on distant peers.}, then replicas of the artifact are sent to them. Else\footnote{The artifact is complete (it no longer contains buds), or semi-complete (it contains buds that were created on other peers and on which, the actor on the current peer is not accredited in writing).}, a replica is returned to the peer from whom the artifact was previously received.
\begin{figure}[ht!]
	\noindent
	\makebox[\textwidth]{\includegraphics[scale=0.15]{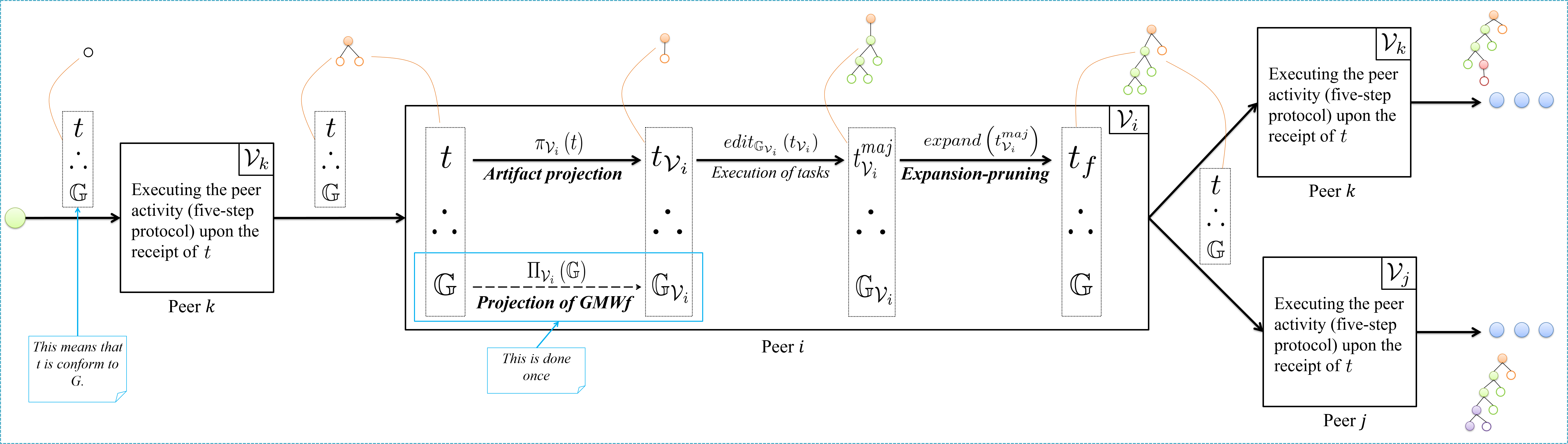}}
	\caption{Overview of the distributed execution of a given process.}
	\label{fig:systemOverview}
\end{figure}

\section{Projection Algorithms for the Distributed Execution of GMAWfP}
\label{sec:projectionAlgorithms}
The GMAWfP execution model is mainly based on three algorithms: \textit{artifact projection}, \textit{GMWf projection} and \textit{expansion}. In this section, we propose versions of these algorithms as well as a study of some of their mathematical properties guaranteeing the correction of the execution model, and the coherence of the distributed system formed by the peers in charge of the execution of a given GMAWfP.

\subsection{The Artifact Projection Algorithm}
\subsubsection{The Algorithm}
Technically, the projection $t_{\mathcal{V}_i}$ of an artifact $t$ according to the view $\mathcal{V}_{i} = \mathcal{A}_{A_i(r)}$ is obtained by deleting in $t$ all nodes whose types do not belong to $\mathcal{V}_{i}$ (all invisible nodes). In our case, the main challenges in this operation are:
\begin{enumerate}
	\item[\textbf{(1)}] nodes of $t_{\mathcal{V}_i}$ must preserve the previously existing execution order between them in $t$,
	\item[\textbf{(2)}] $t_{\mathcal{V}_i}$ must be build by using exclusively the only two forms of production retained for GMWf and
	\item[\textbf{(3)}] $t_{\mathcal{V}_i}$ must be unique in order to ensure the continuation of process execution (see sec. \ref{sec:peerActivity}).
\end{enumerate}

The projection operation is noted $\pi$. Inspired by the one proposed in \cite{badouelTchoupeCmcs}, it projects an artifact by preserving the hierarchy (father-son relationship) between nodes of the artifact (it thus meets challenge \textbf{(1)}); but in addition, it inserts into the projected artifact when necessary, new additional \textit{(re)structuring symbols} (accessible in reading and writing by the agent for whom the projection is made). This enables it to meet challenge \textbf{(2)}. The details of how to accomplish the challenge \textbf{(3)} are outlined immediately after the algorithm (algorithm \ref{algo:artifact-projection}) is presented.
\begin{figure}[ht!]
	\noindent
	\makebox[\textwidth]{\includegraphics[scale=0.24]{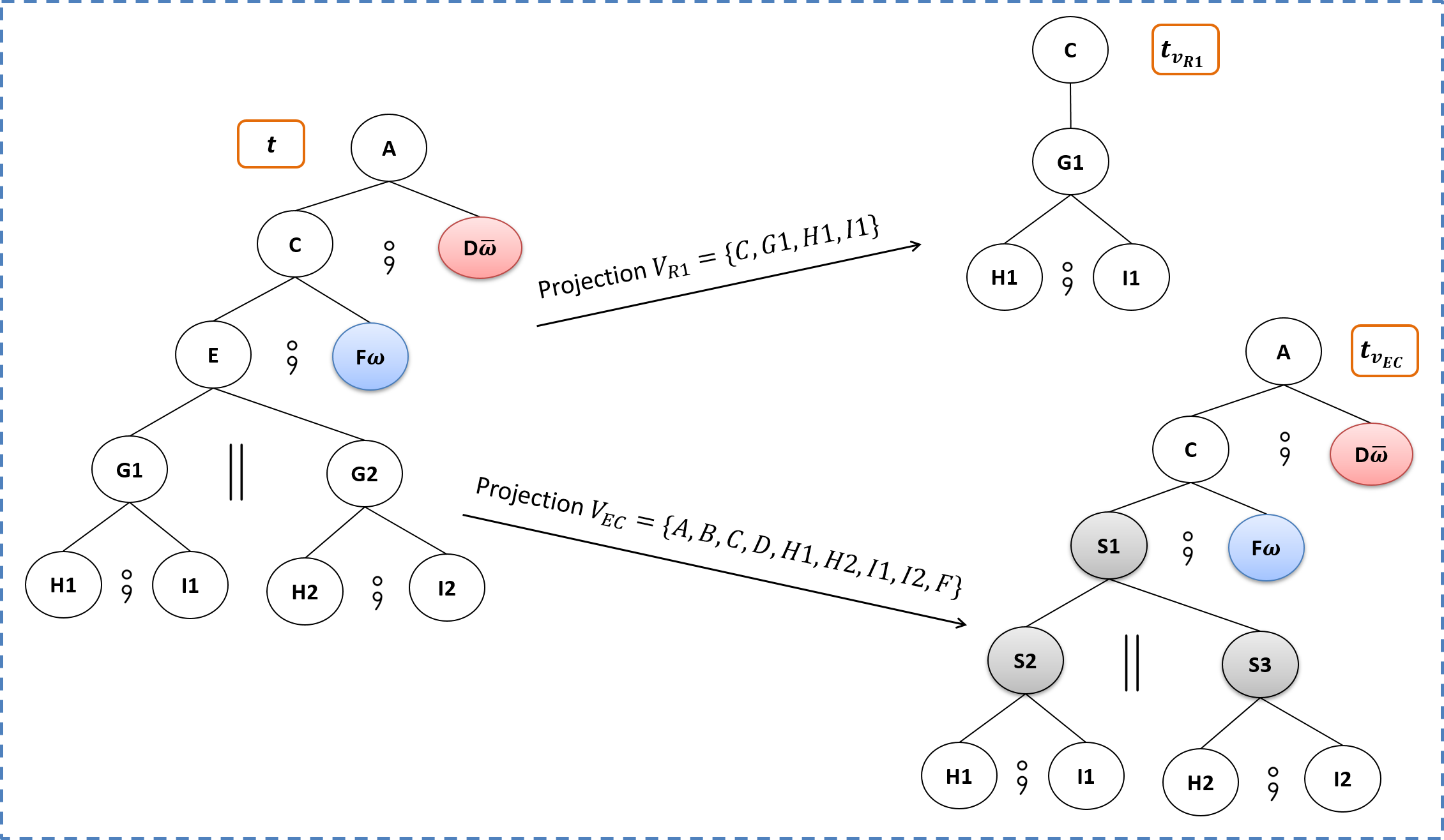}}
	\caption{Example of projections made on an artifact and partial replicas obtained.}
	\label{fig:partial-replicas}
\end{figure}

Figure \ref{fig:partial-replicas} illustrates the projection of an artifact of the peer-review process relatively to the $R1$ (first referee) and $EC$ (Editor in Chief) views. Note the presence in $t_{\mathcal{V}_{EC}}$ of new (re)structuring symbols (in gray). These last ones make it possible to avoid introducing in $t_{\mathcal{V}_{EC}}$, the production $p: C \rightarrow H1 \fatsemi I1 \parallel H2 \fatsemi I2 \fatsemi F$ whose form does not correspond to the two forms of production retained for the GMWf writing\footnote{Note that this production specifies in its right-hand side that we must have parallel and sequential treatments.
Inserting $S1$, $S2$ and $S3$ allows to rewrite $p$ in four productions $p1: C \rightarrow S1 \fatsemi F$, $p2: S1 \rightarrow S2 \parallel S3$, $p3: S2 \rightarrow H1 \fatsemi I1$ and $p4: S3 \rightarrow H2 \fatsemi I2$.}.

Let's consider an artifact $t$ and note by $n=X\left[t_1,\ldots,t_m\right]$ a node of $t$ labelled with the symbol $X$ and having $m$ sub-artifacts $t_1,\ldots,t_m$. Note also by $p_n$, the production of the GMWf that was used to extend node $n$; the type of $p_n$ is either \textit{sequential} (i.e. $p_n$ is of the form $p_n: X \rightarrow X_1 \fatsemi \ldots \fatsemi X_m$ where $X_1,\ldots,X_m$ are the roots of the sub-artifacts $t_1,\ldots,t_m$) or \textit{parallel} ($p_n: X \rightarrow X_1 \parallel \ldots \parallel X_m$). 
Concretely, to project $t$ according to a given view $\mathcal{V}$ (i.e to find $\mathit{projs_t}=\pi_{\mathcal{V}}\left(t \right)$), a depth path of $t$ is performed and invisible nodes are erased or replaced by new nodes associated with (re)structuring symbols to preserve the subtree structure. To do so, the recursive processing presented in algorithm \ref{algo:artifact-projection}, is applied to the root node $n=X\left[t_1,\ldots,t_m\right]$ of $t$.

\begin{algorithm}
	\small
	\caption{Algorithm to project a given artifact according to a given view.}
	\label{algo:artifact-projection}
	\begin{mdframed}[style=MyFrame]
		{\large\textbullet} $~$ \textbf{If symbol $X$ is visible ($X \in \mathcal{V}$)} then :
		
		\textbf{1.}$~$ $n$ is kept in the artifact;
		
		\textbf{2.}$~$ For each sub-artifact $t_i$ of $n$, having node $n_i=X_i\left[t_{i_1},\ldots,t_{i_k}\right]$ as root (of which $p_{n_i}$ is the production that was used to extend it), the following processing is applied :
		
		$~~$\textbf{a.}$~$ The projection of $t_i$ according to $\mathcal{V}$ is done. We obtain the list $\mathit{projs_{t_i}} = \pi_{\mathcal{V}}\left(t_i \right) = \left\{t_{i_{\mathcal{V}_1}},\ldots,t_{i_{\mathcal{V}_l}}\right\}$;
		
		$~~~~$\textbf{b.}$~$ If the type of $p_{n_i}$ is the same as the type of $p_{n}$ or the projection of $t_i$ has produced no more than one artifact ($\left|\mathit{projs_{t_i}}\right| \leq 1$), we just replace $t_i$ by artifacts $t_{i_{\mathcal{V}_1}},\ldots,t_{i_{\mathcal{V}_l}}$ of the list $\mathit{projs_{t_i}}$;
		
		$~~~~$ Otherwise, a new (re)structuring symbol $S_i$ is introduced and we replace the sub-artifact $t_i$ with a new artifact $new\_t_i$ whose root node is $n_{t_i}=S_i\left[t_{i_{\mathcal{V}_1}},\ldots,t_{i_{\mathcal{V}_l}}\right]$;
		
		\textbf{3.}$~$ If the list of new sub-artifacts of $n$ contains only one element $t_1$ having $n_1=S_1\left[t_{1_{\mathcal{V}_1}},\ldots,t_{1_{\mathcal{V}_l}}\right]$ (with $S_1$ a newly created (re)structuring symbol) as root node, we replace in this one, $t_1$ by the sub-artifacts $t_{1_{\mathcal{V}_1}},\ldots,t_{1_{\mathcal{V}_l}}$ of $n_1$. This removes a non-important (re)structuring symbol $S_1$.
		
		\noindent{\large\textbullet} $~$ \textbf{Else}, $n$ is deleted and the result of the projection ($\mathit{projs_t}$) is the union of the projections of each of its sub-artifacts: $\mathit{projs_t} = \pi_{\mathcal{V}}\left(t \right) = \bigcup^{m}_{i=1} \pi_{\mathcal{V}}\left(t_i \right)$
	\end{mdframed}
\end{algorithm}

Note that the algorithm described here applies to all artifacts (including those containing buds) because there is no need to apply any special treatment (locking or unlocking) on buds: they must also be just erased or kept in the artifact to ensure consistency of the execution model. Note also that, this algorithm can return several artifacts (a forest). To avoid that it produces a forest in some cases and thus meet challenge \textbf{(3)}, we make the following assumption: 
\begin{displayquote}
\textit{GMAWfP manipulated in this work are such that all actors are accredited in reading on the GMWf axioms (\textbf{axioms' visibility assumption}).}
\end{displayquote}
The designer must therefore ensure that all actors are accredited in reading on all GMWf axioms. To do this, after modelling a process $\mathcal{P}_{ad}$ and obtaining its GMWf $\mathbb{G}=\left(\mathcal{S},\mathcal{P},\mathcal{A}\right)$, it is sufficient (if necessary) to create a new axiom $A_{\mathbb{G}}$ on which, all actors will be accredited in reading, and to associate it with new unit productions\footnote{A production of a context free grammar is a \textit{unit production} if it is on the form $A \rightarrow B$, where $A$ and $B$ are non-terminal symbols.} $pa : A_{\mathbb{G}} \rightarrow X_a$ where, $X_a \in \mathcal{A}$ is a symbol labelling the root of a target artifact.
Moreover, the designer of the GMWf must statically choose the actor responsible for initiating the process. This actor will therefore be the only one to possess an accreditation in writing on the new axiom $A_{\mathbb{G}}$.

\subsubsection{Studying the Stability Property with the Artifact Projection Algorithm}
\begin{proposition}
	\label{propositionStabiliteProjArt}
	For all GMAWfP $\mathbb{W}_f=\left(\mathbb{G}, \mathcal{L}_{P_k}, \mathcal{L}_{\mathcal{A}_k} \right)$ verifying the axioms' visibility assumption, the projection of an artifact $t$ which is conform to its GMWf ($t \therefore \mathbb{G}$) according to a given view $\mathcal{V}$, results in a single artifact $t_{\mathcal{V}}=\pi_{\mathcal{V}} \left(t\right)$ (stability property of artifacts through the usage of $\pi$).
\end{proposition}

\begin{proof}
	Let's show that $\pi_{\mathcal{V}} \left(t\right)$ produces a single tree $t_{\mathcal{V}}$ which is an artifact. 
	Note that the only case in which the projection of an artifact $t$ according to a view $\mathcal{V}$ produces a forest, is when the root node of $t$ is associated with an invisible symbol $X$ ($X \notin \mathcal{V}$). Knowing that $t \therefore \mathbb{G}$ and that $\mathbb{W}_f$ validates the axioms' visibility assumption, it is deduced that the root node of $t$ is labelled by one of the axioms $A_{\mathbb{G}}$ of $\mathbb{G}$ and that $A_{\mathbb{G}} \in \mathcal{V}$ (hence the uniqueness of the produced tree). Since the projection operation preserves the form of productions, it is concluded that $t_{\mathcal{V}}=\pi_{\mathcal{V}} \left(t\right)$ is an artifact.
\end{proof}

\subsection{The GMWf Projection Algorithm}
\subsubsection{The Algorithm}
The goal of this algorithm is to derive by projection of a given GMWf $\mathbb{G}=\left(\mathcal{S},\mathcal{P},\mathcal{A}\right)$ according to a view $\mathcal{V}$, a local GMWf $\mathbb{G}_{\mathcal{V}} = \left(\mathcal{S}_{\mathcal{V}},\mathcal{P}_{\mathcal{V}}, \mathcal{A}_{\mathcal{V}}\right)$ (we note $\mathbb{G}_{\mathcal{V}} = \Pi_{\mathcal{V}}\left(\mathbb{G} \right)$). The proposed algorithm (algorithm \ref{algo:gmwf-projection}) generates the set of target artifacts\footnote{This generation is necessary because each peer is only configured using the GMWf $\mathbb{G}$ and therefore does not possess all its target artifacts, even though the designer produced $\mathbb{G}$ using these artifacts.} denoted by $\mathbb{G}$ then, it simply project each target artifact according to the view $\mathcal{V}$, then gather the productions in the obtained partial replicas while removing the duplicates.

\begin{algorithm}
	\small
	\caption{Algorithm to project a given GMWf according to a given view.}
	\label{algo:gmwf-projection}
	\begin{mdframed}[style=MyFrame]
		
		\noindent\textbf{1.}$~$ First of all, it is necessary to generate all the target artifacts denoted by $\mathbb{G}$ (see note (1) below); 
		we thus obtain a set $arts_{\mathbb{G}}=\left\{t_1,\ldots,t_n\right\}$;
		
		\noindent\textbf{2.}$~$ Then, each of the target artifacts must be projected according to $\mathcal{V}$. We thus obtain a set $arts_{\mathbb{G}_{\mathcal{V}}} = \left\{t_{\mathcal{V}_1},\ldots,t_{\mathcal{V}_m}\right\}$ (with $m \leq n$ because there may be duplicates; 
		in this case, only one copy is kept) of artifacts partial replicas;
		
		\noindent\textbf{3.}$~$ Then, collect the different (re)structuring symbols appearing in artifacts of $arts_{\mathbb{G}_{\mathcal{V}}}$, making sure to remove duplicates (see note (2) below) 
		and to accordingly update the artifacts and the set $arts_{\mathbb{G}_{\mathcal{V}}}$. We thus obtain a set $\mathcal{S}_{\mathcal{V}_{Struc}}$ of symbols and a final set $arts_{\mathbb{G}_{\mathcal{V}}} = \left\{t_{\mathcal{V}_1},\ldots,t_{\mathcal{V}_l}\right\}$ (with $l \leq m$) of artifacts. These are exactly the only ones that must be conform to the searched GMWf $\mathbb{G}_{\mathcal{V}}$. So we call them, \textit{local target artifacts for the view $\mathcal{V}$};
		
		\noindent\textbf{4.}$~$ At this stage, it is time to collect all the productions that made it possible to build each of the \textit{local target artifacts for the view $\mathcal{V}$}. We obtain a set $\mathcal{P}_{\mathcal{V}}$ of distinct productions.\\
		\textbf{The searched local GMWf $\mathbb{G}_{\mathcal{V}} = \left(\mathcal{S}_{\mathcal{V}},\mathcal{P}_{\mathcal{V}}, \mathcal{A}_{\mathcal{V}}\right)$ is such as}:
		
		$~~$\textbf{a.}$~$ its set of symbols is $\mathcal{S}_{\mathcal{V}} = \mathcal{V} \cup \mathcal{S}_{\mathcal{V}_{Struc}}$;
		
		$~~$\textbf{b.}$~$ its set of productions is $\mathcal{P}_{\mathcal{V}}$;
		
		$~~$\textbf{c.}$~$ its axioms are in $\mathcal{A}_{\mathcal{V}} = \mathcal{A}$
		
		~
		
		\noindent\textit{\textbf{Note (1):}$~$ To generate all the target artifacts denoted by a GMWf $\mathbb{G}=\left(\mathcal{S},\mathcal{P},\mathcal{A}\right)$, one just has to use the set of productions to generate the set of artifacts having one of the axiom $A_{\mathbb{G}}\in \mathcal{A}$ as root. In fact, for each axiom $A_{\mathbb{G}}$, it should be considered that every $A_{\mathbb{G}}$-production $P=\left(A_{\mathbb{G}},X_1\cdots X_n\right)$ induces artifacts $\left\{t_1, \ldots, t_m\right\}$ such as: the root node of each $t_i$ is labelled $A_{\mathbb{G}}$ and has as its sons, a set of artifacts $\left\{t_{i_1},\ldots,t_{i_n}\right\}$, part of the Cartesian product of the sets of artifacts generated when considering each symbol $X_1,\cdots, X_n$ as root node.}
		
		\noindent\textit{\textbf{Note (2):}$~$ In this case, two (re)structuring symbols are identical if for all their appearances in nodes of the different artifacts of $arts_{\mathbb{G}_{\mathcal{V}}}$, they induce the same local scheduling.}
	\end{mdframed}
\end{algorithm}

Figure \ref{fig:gmwf-projection} illustrates the research of a local model $\mathbb{G}_{\mathcal{V}_{EC}}$ such as $\mathbb{G}_{\mathcal{V}_{EC}} = \Pi_{\mathcal{V}_{EC}}\left(\mathbb{G}\right)$ with $\mathcal{V}_{EC}=\mathcal{A}_{EC(r)}=\{A, B, C, D, H1, H2, I1, I2, F\}$. Target artifacts generated from $\mathbb{G}$ (fig. \ref{fig:gmwf-projection}(b)) are projected to obtain two \textit{local target artifacts for the view $\mathcal{V}_{EC}$} (fig. \ref{fig:gmwf-projection}(c)). 
From the local target artifacts thus obtained, the searched GMWf is produced (fig. \ref{fig:gmwf-projection}(d)).
\begin{figure}[ht!]
	\noindent
	\makebox[\textwidth]{\includegraphics[scale=0.24]{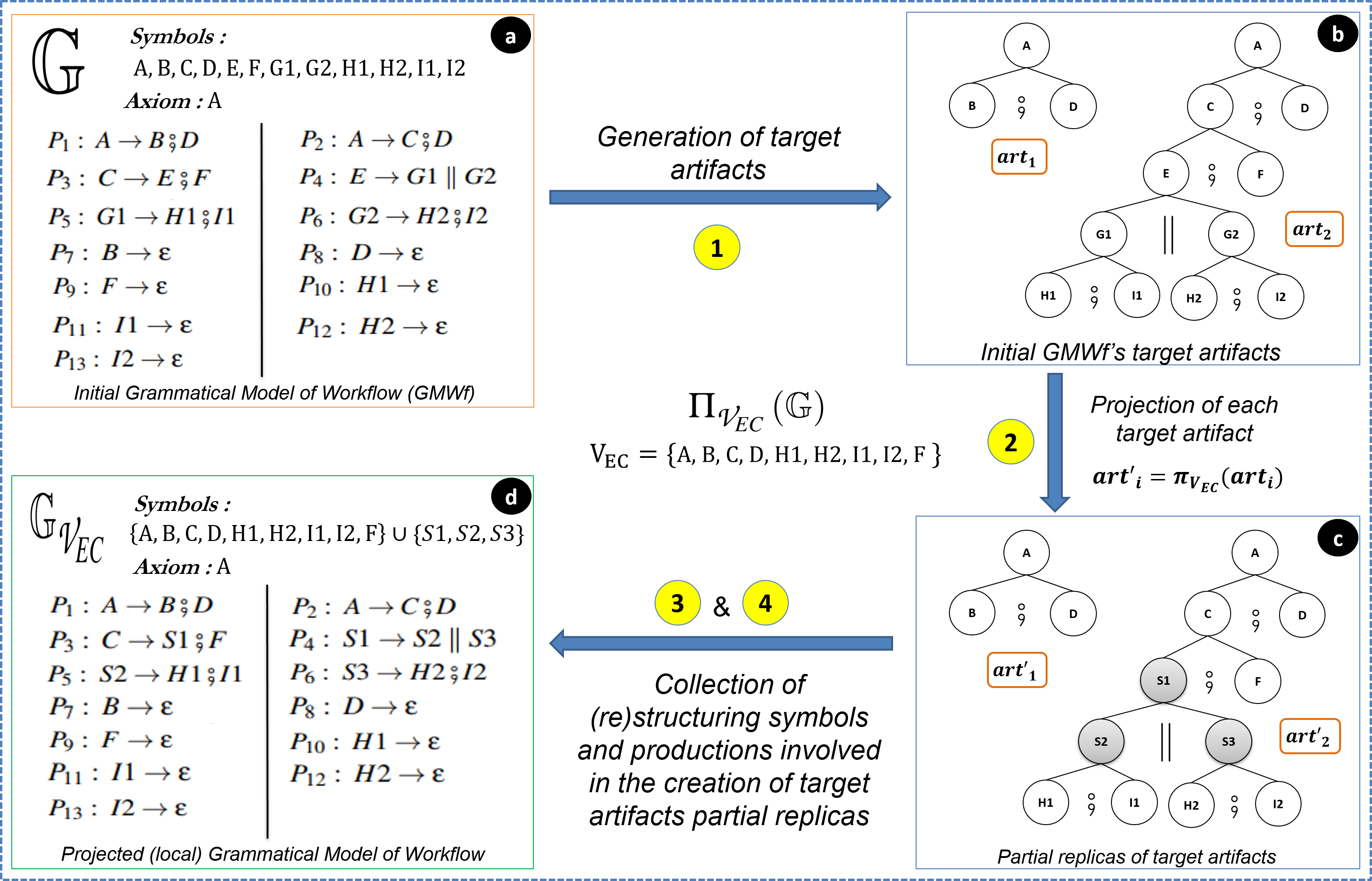}}
	\caption{Example of projection of a GMWf according to a given view.}
	\label{fig:gmwf-projection}
\end{figure}

The GMWf projection algorithm presented here only works for GMWf that do not allow recursive symbols\footnote{It is only in this context that all the target artifacts can be enumerated.}. We therefore assume that:
\begin{displayquote}
	\textit{For the execution model presented in this paper, the manipulated GMAWfP are those whose GMWf do not contain recursive symbols (\textbf{non-recursive GMWf assumption})}.
\end{displayquote} 
With this assumption, it is no longer possible to express iterative routing between process tasks (in the general case); except in cases where the maximum number of iterations is known in advance.

\subsubsection{Some Properties of the GMWf Projection Algorithm}
\begin{proposition}
	\label{propositionStabiliteProjGMWf}
	For all GMAWfP $\mathbb{W}_f=\left(\mathbb{G}, \mathcal{L}_{P_k}, \mathcal{L}_{\mathcal{A}_k} \right)$ verifying the axioms' visibility and the non-recursivity of GMWf assumptions, the projection of its GMWf $\mathbb{G}=\left(\mathcal{S},\mathcal{P},\mathcal{A}\right)$ according to a given view $\mathcal{V}$, is a GMWf $\mathbb{G}_{\mathcal{V}} = \Pi_{\mathcal{V}}\left(\mathbb{G} \right)$ for a GMAWfP $\mathbb{W}_{f_{\mathcal{V}}}$ verifying the assumptions of axiom visibility and non-recursivity of GMWf (stability property of GMWf through the usage of $\Pi$).
\end{proposition}

\begin{proof}
	Let's show that $\mathbb{G}_{\mathcal{V}} = \Pi_{\mathcal{V}}\left(\mathbb{G} \right)$ is a GMWf for a new GMAWfP $\mathbb{W}_{f_{\mathcal{V}}}=\left(\mathbb{G}_{\mathcal{V}}, \mathcal{L}_{P_k}, \mathcal{L}_{\mathcal{A}_{\mathcal{V}_k}} \right)$ that verifies the assumptions of axioms' visibility and non-recursivity of GMWf. 
	As $\mathbb{W}_f=\left(\mathbb{G}, \mathcal{L}_{P_k}, \mathcal{L}_{\mathcal{A}_k} \right)$ validates the non-recursivity of GMWf assumption, the set of target artifacts ($arts_{\mathbb{G}}=\left\{t_1,\ldots,t_n\right\}$) that it denotes is finite and can therefore be fully enumerated. Knowing further that $\mathbb{W}_f$ validates the axioms' visibility assumption, it is deduced that the set $arts_{\mathbb{G}_{\mathcal{V}}} = \left\{t_{\mathcal{V}_1}=\pi_{\mathcal{V}}\left(t_1\right), \ldots,t_{\mathcal{V}_n}=\pi_{\mathcal{V}}\left(t_n\right)\right\}$ is finite and the root node of each artifact $t_{\mathcal{V}_i}$ is associated with an axiom $A_{\mathbb{G}} \in \mathcal{A}$ (see proposition \ref{propositionStabiliteProjArt}). $\mathbb{G}_{\mathcal{V}}$ being built from the set $arts_{\mathbb{G}_{\mathcal{V}}}$, its axioms $\mathcal{A}_{\mathcal{V}}=\mathcal{A}$ are visible to all actors and its productions are only of the two forms retained for GMWf. In addition, each new (re)structuring symbol ($S \in \mathcal{S}_{\mathcal{V}_{Struc}}$)) is created and used only once to replace a symbol that is not visible and not recursive (by assumption) when projecting artifacts of $arts_{\mathbb{G}}$. The new symbols are therefore not recursive. By replacing in $\mathcal{L}_{\mathcal{A}_k}$ the view $\mathcal{V}$ by $\mathcal{V} \cup \mathcal{S}_{\mathcal{V}_{Struc}}$, one obtains a new set $\mathcal{L}_{\mathcal{A}_{\mathcal{V}_k}}$ of accreditations for a new GMAWfP $\mathbb{W}_{f_{\mathcal{V}}}=\left(\mathbb{G}_{\mathcal{V}}, \mathcal{L}_{P_k}, \mathcal{L}_{\mathcal{A}_{\mathcal{V}_k}} \right)$ verifying the assumptions of axioms' visibility and non-recursivity of GMWf.
\end{proof}

\begin{proposition}
	\label{propositionCoherenceArtefact}
	For all GMAWfP $\mathbb{W}_f=\left(\mathbb{G}, \mathcal{L}_{P_k}, \mathcal{L}_{\mathcal{A}_k} \right)$ verifying the axioms' visibility and the non-recursivity of GMWf assumptions, the projection of an artifact $t$ which is conform to the GMWf $\mathbb{G}$ according to a given view $\mathcal{V}$, is an artifact which is conform to the projection of $\mathbb{G}$ according to $\mathcal{V}$ $\left(\forall t, ~t \therefore \mathbb{G} \Rightarrow \pi_{\mathcal{V}}\left(t\right) \therefore \Pi_{\mathcal{V}}\left(\mathbb{G} \right)\right)$.
\end{proposition}

\begin{proof}
	Knowing that the considered GMAWfP $\mathbb{W}_f=\left(\mathbb{G}, \mathcal{L}_{P_k}, \mathcal{L}_{\mathcal{A}_k} \right)$ verifies the axioms' visibility and the non-recursivity of GMWf assumptions, it is deduced that the set of its target artifacts $arts_{\mathbb{G}}$ (those who helped to build its GMWf $\mathbb{G}$) is finite and any artifact that is conform to its GMWf $\mathbb{G}$ is a target artifact $\left( \forall t, ~t \therefore \mathbb{G} \Leftrightarrow t \in arts_{\mathbb{G}} \right)$. Therefore, considering a given artifact $t$ such that $t$ is conform to $\mathbb{G}$ ($t \therefore \mathbb{G}$), one knows that it is a target artifact ($t \in arts_{\mathbb{G}}$) and its projection according to a given view $\mathcal{V}$ produces a single artifact $t_{\mathcal{V}}=\pi_{\mathcal{V}}\left(t\right)$ (see "stability property of $\pi$", proposition \ref{propositionStabiliteProjArt}) such as $t$ and $t_{\mathcal{V}}$ have the same root (one of the axioms $A_{\mathbb{G}} \in \mathcal{A}$ of $\mathbb{G}$). Since $t$ is a target artifact, its projection $t_{\mathcal{V}}$ (through the renaming of some potential (re)structuring symbols) is part of the set $arts_{\mathbb{G}_{\mathcal{V}}}$ of artifacts that have generated $\mathbb{G}_{\mathcal{V}} = \Pi_{\mathcal{V}}\left(\mathbb{G} \right)$ by applying the projection principle described in the algorithm \ref{algo:gmwf-projection}. Therefore, the productions involved in the construction of $t_{\mathcal{V}}$ are all included in the set of productions of the GMWf $\mathbb{G}_{\mathcal{V}} = \Pi_{\mathcal{V}}\left(\mathbb{G} \right)$. As the set of axioms of $\mathbb{G}_{\mathcal{V}}$ is $\mathcal{A}_{\mathcal{V}} = \mathcal{A}$, it is deduced that $A_{\mathbb{G}} \in \mathcal{A}_{\mathcal{V}}$ and concluded that $t_{\mathcal{V}} \therefore \mathbb{G}_{\mathcal{V}}$.
\end{proof}

\begin{proposition}
	\label{propositionReciproqueCoherenceArtefact}
	Consider a GMAWfP $\mathbb{W}_f=\left(\mathbb{G}, \mathcal{L}_{P_k}, \mathcal{L}_{\mathcal{A}_k} \right)$ verifying the axioms' visibility and the non-recursivity assumptions. For all artifact $t_{\mathcal{V}}$ which is conform to $\Pi_{\mathcal{V}}\left(\mathbb{G} \right)$, it exists at least one artifact $t$ which is conform to $\mathbb{G}$ such that $t_{\mathcal{V}}=\pi_{\mathcal{V}}\left(t\right)$ $\left(\forall t_{\mathcal{V}}, ~t_{\mathcal{V}} \therefore \Pi_{\mathcal{V}}\left(\mathbb{G} \right) \Rightarrow \exists t, ~t \therefore \mathbb{G} ~and~ t_{\mathcal{V}}=\pi_{\mathcal{V}}\left(t\right) \right)$.
\end{proposition}

\begin{proof}
	With proposition \ref{propositionStabiliteProjGMWf} ("stability property of $\Pi$") it has been shown that the projection $\mathbb{G}_{\mathcal{V}} = \Pi_{\mathcal{V}}\left(\mathbb{G} \right)$ according to the view $\mathcal{V}$ of a GMWf $\mathbb{G}$ verifying the axioms' visibility and the non-recursivity assumptions, is a GMWf verifying the same assumptions. On this basis and using similar reasoning to that used to prove the proposition \ref{propositionCoherenceArtefact}, it's been determined that an artifact $t_{\mathcal{V}}$ that is conform to $\mathbb{G}_{\mathcal{V}}$, is one of its target artifacts (\textit{local target artifact for the view $\mathcal{V}$}): i.e, $t_{\mathcal{V}} \in arts_{\mathbb{G}_{\mathcal{V}}}$. Referring to the projection process which made it possible to obtain $\mathbb{G}_{\mathcal{V}}$, it is determined that the set $arts_{\mathbb{G}_{\mathcal{V}}}$ is exclusively made up of the projections of the set $arts_{\mathbb{G}}=\left\{t_1,\ldots,t_n\right\}$ of $\mathbb{G}$'s target artifacts. $t_{\mathcal{V}}$ is therefore the projection of at least one target artifact $t_i \in arts_{\mathbb{G}}$ of $\mathbb{G}$ $\left(t_{\mathcal{V}}=\pi_{\mathcal{V}}\left(t_i\right)\right)$. Knowing that $\forall t, ~t \therefore \mathbb{G} \Leftrightarrow t \in arts_{\mathbb{G}}$ (see proof of proposition \ref{propositionCoherenceArtefact}), it is deduced that $t_i \therefore \mathbb{G}$ and the proof of this proposition is made.
\end{proof}

By applying the GMWf projection algorithm presented above (algorithm \ref{algo:gmwf-projection}) to the running example, one obtain the productions listed in table \ref{tableau:gramLocales} for the different actors respectively. In the illustrated case here, we have considered an update of the GMWf of the peer-review process so that it validates the axioms' visibility assumption.
\begin{table}[h]
	\centering
	\caption{Local GMWf productions of all the actors involved in the peer-review process.}
	\label{tableau:gramLocales}
	\begin{tabular}[t]{|m{1.1cm}|m{10.5cm}|}
		\hline
		\textbf{Actor} & \textbf{Productions of local GMWf} \\
		\hline
		$EC$ &
		\[ 
		\begin{array}{l|l|l}
		P_{1}:\; A_{\mathbb{G}}\rightarrow A \;\; & \;\; P_{2}:\; A\rightarrow B\fatsemi D\;\; & \;\; P_{3}:\; A\rightarrow C\fatsemi D  \\
		P_{4}:\; C\rightarrow S1\fatsemi F \;\; & \;\; P_{5}:\; S1\rightarrow S2\parallel S3\;\; & \;\; P_{6}:\; S2\rightarrow H1 \fatsemi I1  \\
		P_{7}:\; S3\rightarrow H2 \fatsemi I2 \;\; & \;\; P_{8}:\; B\rightarrow \varepsilon\;\; & \;\; P_{9}:\; D\rightarrow \varepsilon \\
		P_{10}:\; F\rightarrow \varepsilon \;\; & \;\; P_{11}:\; H1\rightarrow \varepsilon\;\; & \;\; P_{12}:\; I1\rightarrow \varepsilon  \\
		P_{13}:\; H2\rightarrow \varepsilon \;\; & \;\; P_{14}:\; I2\rightarrow \varepsilon \;\; & \;\;  \\
		\end{array}
		\]
		\\
		\hline
		$AE$ & 
		\[ 
		\begin{array}{l|l|l}
		P_{1}:\; A_{\mathbb{G}}\rightarrow A \;\; & \;\; P_{2}:\; A\rightarrow C \;\; & \;\; P_{3}:\; C\rightarrow E\fatsemi F  \\
		P_{4}:\; E\rightarrow S1\parallel S2 \;\; & \;\; P_{5}:\; S1\rightarrow H1\fatsemi I1 \;\; & \;\; P_{6}:\; S2\rightarrow H2\fatsemi I2  \\
		P_{7}:\; H1\rightarrow \varepsilon \;\; & \;\; P_{8}:\; I1\rightarrow \varepsilon \;\; & \;\; P_{9}:\; H2\rightarrow \varepsilon \\
		P_{10}:\; I2\rightarrow \varepsilon \;\; & \;\; P_{11}:\; F\rightarrow \varepsilon \;\; & \;\; P_{12}:\; A_{\mathbb{G}}\rightarrow \varepsilon \\
		\end{array}
		\]
		\\
		\hline
		$R1$ & 
		\[ 
		\begin{array}{l|l|l}
		P_{1}:\; A_{\mathbb{G}}\rightarrow C \;\; & \;\; P_{2}:\; C\rightarrow G1\;\; & \;\; P_{3}:\; G1\rightarrow H1\fatsemi I1 \\
		P_{4}:\; H1\rightarrow \varepsilon \;\; & \;\; P_{5}:\; I1\rightarrow \varepsilon \;\; & \;\; P_{6}:\; A_{\mathbb{G}}\rightarrow \varepsilon \\
		\end{array}
		\]
		\\
		\hline
		$R2$ & 
		\[ 
		\begin{array}{l|l|l}
		P_{1}:\; A_{\mathbb{G}}\rightarrow C \;\; & \;\; P_{2}:\; C\rightarrow G2\;\; & \;\; P_{3}:\; G2\rightarrow H2\fatsemi I2 \\
		P_{4}:\; H2\rightarrow \varepsilon \;\; & \;\; P_{5}:\; I2\rightarrow \varepsilon \;\; & \;\; P_{6}:\; A_{\mathbb{G}}\rightarrow \varepsilon \\
		\end{array}
		\]
		\\
		\hline
	\end{tabular}
\end{table}

\subsection{The Expansion Algorithm}
\subsubsection{The Algorithm}
Consider an (global) artifact under execution $t$, and $t_{\mathcal{V}}=\pi_{\mathcal{V}}\left(t\right)$ its partial replica on the site of an actor $A_i$ whose view is $\mathcal{V}$. Consider the partial replica $t_{\mathcal{V}}^{maj} \geq t_{\mathcal{V}}$ obtained by developing some unlocked buds of $t_{\mathcal{V}}$ as a result of $A_i$'s contribution. The expansion problem consists in finding an (global) artifact under execution $t_f$, which integrates nodes of $t$ and $t_{\mathcal{V}}$. To solve this problem made difficult by the fact that $t$ and $t_{\mathcal{V}}$ are conform to two different models ($\mathbb{G}$ and $\mathbb{G}_{\mathcal{V}} = \Pi_{\mathcal{V}} \left(\mathbb{G} \right)$), we perform a three-way merge {\cite{tomMens}. We merge the artifacts $t$ and $t_{\mathcal{V}}$ using a (global) target artifact $t_g$ such that: 
\begin{enumerate}
	\item[\textbf{(a)}] $t$ is a prefix of $t_g$ ($t \leq t_g$)
	\item[\textbf{(b)}] $t_{\mathcal{V}}^{maj}$ is a prefix of the partial replica of $t_g$ according to $\mathcal{V}$ $\left(t_{\mathcal{V}}^{maj} \leq \pi_{\mathcal{V}}\left(t_g \right)\right)$
\end{enumerate}	
The proposed algorithm proceeds in two steps.

\subsubsection*{Step 1 - Search for the merging guide $t_g$}
The search of a merging guide is done by (algorithm \ref{algo:search-guide}) generating the set of target artifacts denoted by $\mathbb{G}$, then filtering this set to retain only those for which $t$ is a prefix (see the definition of the prefix relationship in algorithm \ref{algo:search-guide}) and $t_{\mathcal{V}}^{maj}$ is a prefix of their projection according to the view $\mathcal{V}$.

\begin{algorithm}
	\small
	\caption{Algorithm to search a merging guide.}
	\label{algo:search-guide}
	\begin{mdframed}[style=MyFrame]
		\noindent\textbf{1.}$~$ First of all, we have to generate the set $arts_{\mathbb{G}}=\left\{t_1,\ldots,t_n\right\}$ of target artifacts denoted by $\mathbb{G}$;
		
		\noindent\textbf{2.}$~$ Then, we must filter this set to retain only the artifacts $t_i$ admitting $t$ as a prefix (criterion \textbf{(a)}) and whose projections according to $\mathcal{V}$ ($t_{i_{\mathcal{V}_j}}$) admit $t_{\mathcal{V}}^{maj}$ as a prefix (criterion \textbf{(b)}). It is said that an artifact $t_a$ (whose root node is $n_a=X_a[t_{a_1},\ldots,t_{a_l}]$) is a prefix of a given artifact $t_b$ (whose root node is $n_b=X_b[t_{b_1},\ldots,t_{b_m}]$) if and only if the root nodes $n_a$ and $n_b$ are of the same types (i.e $X_a=X_b$) and:
		
		$~~$\textbf{a.}$~$ The node $n_a$ is a bud or,
		
		$~~$\textbf{b.}$~$ The nodes $n_a$ and $n_b$ have the same number of sub-artifacts (i.e $l=m$), the same type of scheduling for the sub-artifacts and each sub-artifact $t_{a_i}$ of $n_a$ is a prefix of the sub-artifact $t_{b_i}$ of $n_b$.
		
		\noindent We obtain the set $guides=\left\{t_{g_1},\ldots,t_{g_k}\right\}$ of artifacts that can guide the merging;
		
		\noindent\textbf{3.}$~$ Finally, we randomly select an element $t_g$ from the set $guides$.
	\end{mdframed}
\end{algorithm}

\begin{algorithm}
	\small
	\caption{Three-way merging algorithm.}
	\label{algo:three-way-merge}
	\begin{mdframed}[style=MyFrame]
		\noindent A prefixed depth path of the three artifacts ($t$, $t_{\mathcal{V}}^{maj}$ and $t_g$) is made simultaneously until there is no longer a node to visit in $t_g$. Let $n_{t_i}$ (resp. $n_{t_{\mathcal{V}_j}^{maj}}$ and $n_{t_{g_k}}$) be the node located at the address $w_i$ (resp. $w_j$ and $w_k$) of $t$ (resp. $t_{\mathcal{V}}^{maj}$ and $t_g$) and currently being visited. If nodes $n_{t_i}$, $n_{t_{\mathcal{V}_j}^{maj}}$ and $n_{t_{g_k}}$ are such that (\textbf{processing}):
		
		\noindent\textbf{1.}$~$ $n_{t_{\mathcal{V}_j}^{maj}}$ is associated with a (re)structuring symbol (fig. \ref{fig:expansion-pattern}(d)) then: we take a step forward in the depth path of $t_{\mathcal{V}}^{maj}$ and we resume processing;
		
		\noindent\textbf{2.}$~$ $n_{t_i}$, $n_{t_{\mathcal{V}_j}^{maj}}$ and $n_{t_{g_k}}$ exist and are all associated with the same symbol $X$ (fig. \ref{fig:expansion-pattern}(a) and \ref{fig:expansion-pattern}(b)) then:
		we insert $n_{t_{\mathcal{V}_j}^{maj}}$ (it is the most up-to-date node) into $t_f$ at the address $w_k$; 
		if $n_{t_{\mathcal{V}_j}^{maj}}$ is a bud then we prune (delete sub-artifacts) $t_g$ at the address $w_k$; 
		we take a step forward in the depth path of the three artifacts and we resume processing.
		
		\noindent\textbf{3.}$~$ $n_{t_i}$, $n_{t_{\mathcal{V}_j}^{maj}}$ and $n_{t_{g_k}}$ exist and are respectively associated with symbols $X_i$, $X_j$ and $X_k$ such that $X_k \neq X_i$ and $X_k \neq X_j$ (fig. \ref{fig:expansion-pattern}(e)) then: 
		we add $n_{t_{g_k}}$ in $t_f$ at address $w_k$. This is an upstair bud; 
		we take a step forward in the depth path of $t_g$ and we resume processing.
		
		\noindent\textbf{4.}$~$ $n_{t_i}$ (resp. $n_{t_{\mathcal{V}_j}^{maj}}$) and $n_{t_{g_k}}$ exist and are associated with the same symbol $X$ (fig. \ref{fig:expansion-pattern}(c) and \ref{fig:expansion-pattern}(f)) then: 
		we insert $n_{t_i}$ (resp. $n_{t_{\mathcal{V}_j}^{maj}}$) into $t_f$ at the address $w_k$;
		if $n_{t_i}$ (resp. $n_{t_{\mathcal{V}_j}^{maj}}$) is a bud, we prune $t_g$ at the address $w_k$; 
		we take a step forward in the depth path of the artifacts $t$ (resp. $t_{\mathcal{V}}^{maj}$) and $t_g$, then we resume processing.
	\end{mdframed}
\end{algorithm}

\subsubsection*{Step 2 - Merging $t$, $t_{\mathcal{V}}^{maj}$ and $t_g$}
The problem here is to find an artifact $t_f$ that includes all the contributions already made during the workflow execution. The structure of the searched artifact $t_f$ is the same as that of $t_g$: hence the interest to use $t_g$ as a guide. The merging is carried out by the algorithm \ref{algo:three-way-merge}.  Technically, the three artifacts $t$, $t_{\mathcal{V}}^{maj}$ and $t_g$ are explored in depth simultaneously and a specific treatment is applied according to the configuration of the visited nodes: if the three nodes visited at a given time are mergeable (they are of the same type and some are updates of others) then, the most up-to-date node is retained and added to the resulting artifact; otherwise, the nodes preventing the merge are ignored (pruning).

\begin{figure}[ht!]
	\noindent
	\makebox[\textwidth]{\includegraphics[scale=0.24]{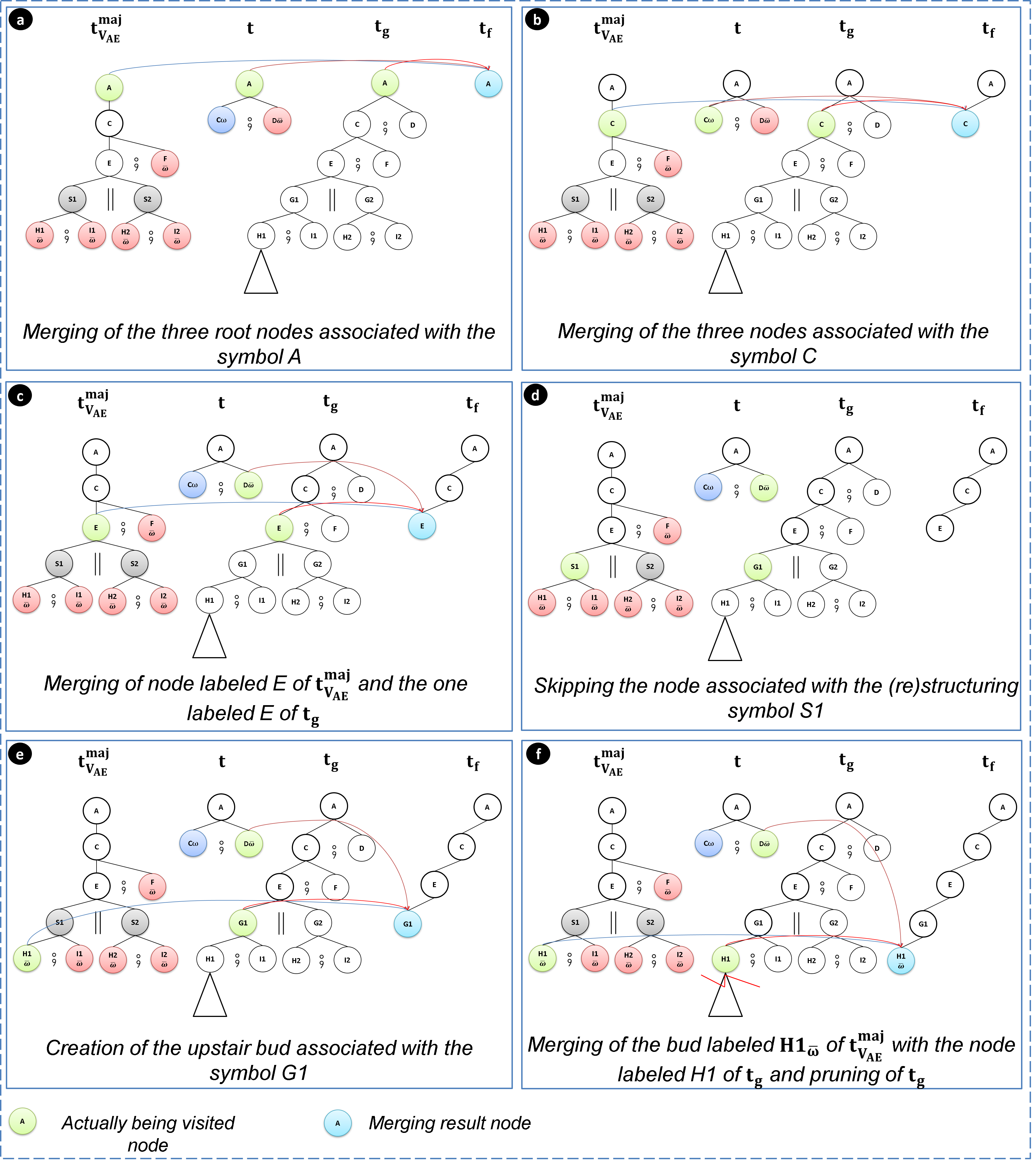}}
	\caption{Some particular cases to be managed during the expansion.}
	\label{fig:expansion-pattern}
\end{figure}

\subsubsection{Some Properties of the Expansion Algorithm}
\begin{proposition}
	\label{propositionAtLeastOneGuide}
	For any update $t_{\mathcal{V}}^{maj}$ in accordance with a GMWf $\mathbb{G}_{\mathcal{V}} = \Pi_{\mathcal{V}}\left(\mathbb{G} \right)$, of a partial replica $t_{\mathcal{V}}=\pi_{\mathcal{V}}\left(t\right)$ obtained by projecting (according to the view $\mathcal{V}$) an artifact $t$ being executed in accordance with the GMWf $\mathbb{G}$ of a GMAWfP verifying the axioms' visibility and the non-recursivity assumptions, there is at least one target artifact (the three-way merge guide) $t_g \in arts_{\mathbb{G}}$ of $\mathbb{G}$ such as:
	\begin{enumerate}
		\item[\textbf{(a)}] $t$ is a prefix of $t_g$ ($t \leq t_g$)
		\item[\textbf{(b)}] $t_{\mathcal{V}}^{maj}$ is a prefix of the partial replica of $t_g$ according to $\mathcal{V}$ $\left(t_{\mathcal{V}}^{maj} \leq \pi_{\mathcal{V}}\left(t_g \right)\right)$
	\end{enumerate}	
\end{proposition}	

\begin{proof}
	Thanks to the proposals \ref{propositionStabiliteProjGMWf}, \ref{propositionCoherenceArtefact} and the artifact editing model used here\footnote{An artifact is developed at the level of its leaves using the productions of the GMWf to which it conforms.}, it is established that since the artifact $t$ being executed in accordance with $\mathbb{G}$ is a prefix of a non-empty set of $\mathbb{G}$'s target artifacts $arts_{\mathbb{G}}^{'} = \left\{t_{1}^{'},\ldots, t_{n}^{'}\right\}$ ($\forall 1 \leq i \leq n, ~t \leq t_{i}^{'}$), its projection $t_{\mathcal{V}}$ according to the view $\mathcal{V}$ is a prefix of a non-empty set $arts_{\mathbb{G}_{\mathcal{V}}}^{'} = \left\{t_{{\mathcal{V}}_{1}}^{'},\ldots, t_{{\mathcal{V}}_{m}}^{'}\right\}$ of $\mathbb{G}_{\mathcal{V}} = \Pi_{\mathcal{V}}\left(\mathbb{G} \right)$'s local target artifacts for the said view ($\forall 1 \leq j \leq m, ~t_{\mathcal{V}} \leq t_{{\mathcal{V}}_{j}}^{'}$): elements of $arts_{\mathbb{G}}^{'}$ are potential merging guides candidates that all verify the property \textbf{(a)}. In addition, using the propositions \ref{propositionStabiliteProjGMWf} and \ref{propositionReciproqueCoherenceArtefact}, it is established that each element of $arts_{\mathbb{G}_{\mathcal{V}}}^{'}$ is the projection of at least one element of $arts_{\mathbb{G}}^{'}$ according to the view $\mathcal{V}$ \textbf{(1)}. Given that $t_{\mathcal{V}}^{maj}$ is obtained by developing buds of $t_{\mathcal{V}}$ in accordance with $\mathbb{G}_{\mathcal{V}}$, it is inferred that $t_{\mathcal{V}}^{maj}$ is a prefix of a non-empty subset $arts_{\mathbb{G}_{\mathcal{V}}}^{maj} \subseteq arts_{\mathbb{G}_{\mathcal{V}}}^{'}$ of local target artifacts for the view $\mathcal{V}$ \textbf{(2)}. With the proposition \ref{propositionReciproqueCoherenceArtefact} once again, it is determined that for each artifact $t_{{\mathcal{V}}_{j}}^{'} \in arts_{\mathbb{G}_{\mathcal{V}}}^{maj}$, there is at least one artifact $t_{g_j}$ that is conform to $\mathbb{G}$ such as $t_{{\mathcal{V}}_{j}}^{'} = \pi_{\mathcal{V}}\left(t_{g_j} \right)$: this new set $arts_{\mathbb{G}}^{maj} = \left\{t_{g_1},\ldots, t_{g_k}\right\}$ is made up of potential merging guides candidates that all verify the property \textbf{(b)}. Results \textbf{(1)} and \textbf{(2)} show that $arts_{\mathbb{G}}^{maj}$ and $arts_{\mathbb{G}}^{'}$ are not disjoint. As a consequence, the set $guides= arts_{\mathbb{G}}^{maj} \cap arts_{\mathbb{G}}^{'}$ of potential merging guides that all verify both property \textbf{(a)} and \textbf{(b)} is not empty.
\end{proof}

\begin{corollary}
	\label{propositionUniqueExpansion}
	For an artifact $t$ being executed in accordance with a GMWf $\mathbb{G}$ of a GMAWfP verifying the axioms' visibility and the non-recursivity assumptions, and an update $t_{\mathcal{V}}^{maj} \geq t_{\mathcal{V}}$ of its partial replica $t_{\mathcal{V}}=\pi_{\mathcal{V}}\left(t\right)$ according to the view $\mathcal{V}$, the expansion of $t_{\mathcal{V}}^{maj}$ contains at least one artifact and the expansion-pruning algorithm presented here returns one and only one artifact.
\end{corollary}

This result (corollary \ref{propositionUniqueExpansion}) derives from the proof of the proposition \ref{propositionAtLeastOneGuide} (\textit{there is always at least one artifact in the expansion of $t_{\mathcal{V}}^{maj}$ under the conditions of corollary \ref{propositionUniqueExpansion}}) and from the fact that in the last instruction of the algorithm \ref{algo:search-guide}, an artifact is randomly selected an returned from a non-empty set of potential guides (\textit{only one of the expansion artifacts is used in the three-way merging}).

\subsection{A Haskell Implementation of the Algorithm presented in this Work}
You can find types and functions (coded in Haskell\footnote{Haskell: \url{https://www.haskell.org/}, visited the 01/08/2020.}) that perform the projections as described in this paper in the \url{https://github.com/MegaMaxim10/GMAWfP-Projection-Algorithms} Git\footnote{Git: \url{https://git-scm.com/}, visited the 01/08/2020.} repository. These tools have been proposed as a proof of concepts. More specifically, they include types for encoding annotations (sequential, parallel, locked, unlocked, etc.) on artifacts, their nodes and productions. There are also simple types for manipulating productions, artifacts, GMWf,... as well as the actual projection functions. All of these tools are provided in a file that can be directly interpreted using a Haskell interpreter like the Glasgow Haskell Compiler\footnote{The Glasgow Haskell Compiler: \url{https://www.haskell.org/ghc/}, visited the 01/08/2020.}. A readme and comments have been added to make it easier to get to grips with the provided implementation as shown in the screenshot in figure \ref{fig:repository}.

\begin{figure}[ht!]
	\noindent
	\makebox[\textwidth]{\includegraphics[scale=0.24]{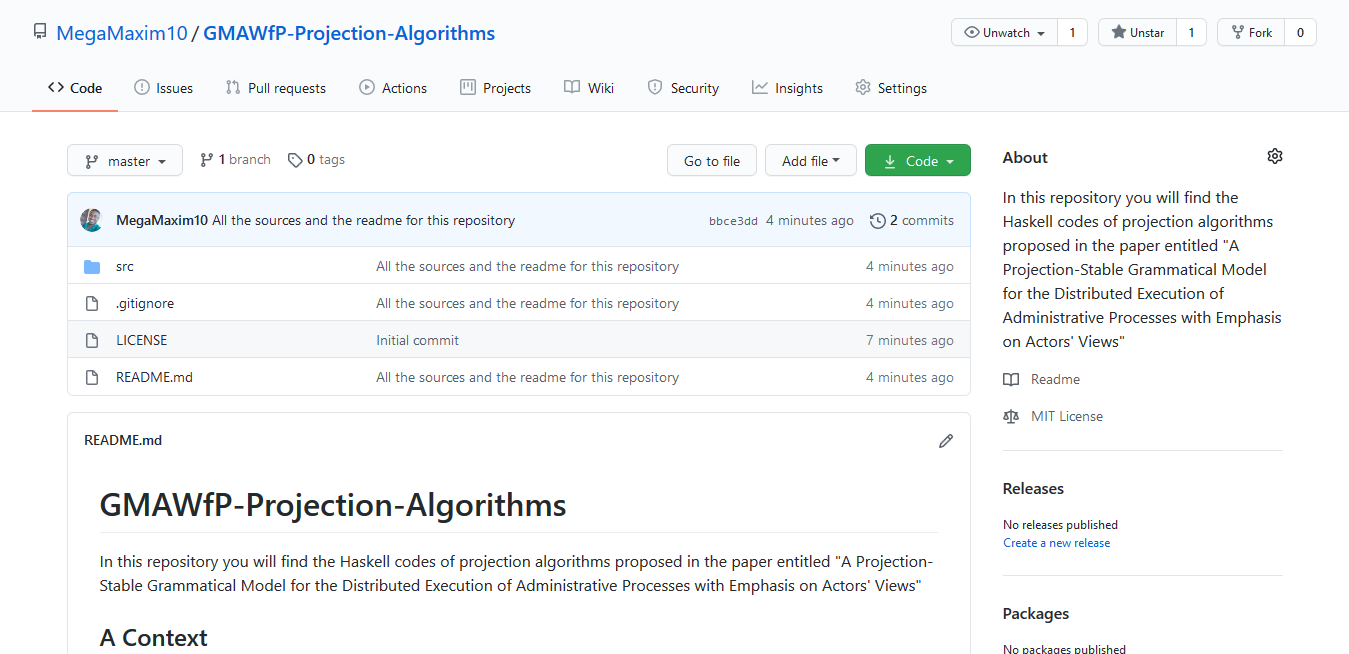}}
	\caption{A screenshot of the provided Git repository.}
	\label{fig:repository}
\end{figure}

One observation that can be quickly made by consulting the provided Haskell code is that, it is quite long. This shows that the proposed algorithms are difficult to present with common notations (in pseudo-code form or even directly in code); this is why we have opted to present them with instructions written in natural language and sprinkled with a few mathematical formulas (it was more concise and precise that way). Nevertheless, we have chosen to use Haskell (a functional language) to code our functions because, the Haskell code is generally self-descriptive (i.e. very close to semi-formal descriptions) and more compact than those written in other languages.

\section{Related Works and Discussion}
\label{Discussion}
\subsection{Projection of Trees in a Cooperative Editing Workflow}
Some works of the literature have focused on the projection of trees that conform to grammatical models in a cooperative editing workflow. We present some of them here and discuss our results as we go along.

In their structured cooperative editing model, Badouel et al. \cite{badouelTchoupeCmcs} proposed a tree projection algorithm operating in the general case: i.e. even in the case in which trees are not annotated and their roots may be invisible; the projection may thus produce a forest. The artifact projection algorithm proposed in this manuscript is a specialization of their own in the case of annotated trees, constructed using only two types of productions: it was therefore necessary to be able to add new (re)structuring symbols during the projection in order to guarantee the stability of models through this operation. Moreover, the scope of application of the algorithm we propose requires that the projection always provides a single artifact: hence the axioms' visibility assumption.

The authors of \cite{badouelTchoupeCmcs} also proposed a solution to the expansion problem. More precisely, they proposed to associate to the updated partial replica whose expansion is sought, a tree automaton with exit states \cite{tchoupeZekeng2017} generating the documents of the expansion. This tree automaton is constructed using the global grammatical model, the considered view and the updated partial replica. In our case, an additional parameter has been added: the global artifact whose partial replica has been updated. This is where the interest of performing a three-way merge comes from. Moreover, we did not need the automaton structure because we made the assumption of non-recursivity of the manipulated grammatical model. Let us also mention that our expansion is followed by a pruning to better correspond to the field of application of this paper (the decentralized execution of GMAWfP).

About our expansion-pruning algorithm, note that the choice of the merging guide is made randomly from a set of potential candidates. This consideration was made to simplify the work presented here. Indeed, if the initial grammatical model is specified without ambiguity (an ambiguity could come from an execution scenario that contains another one) then there is no problem with this consideration. However, in the presence of a grammatical model subject to design errors, the choice of a three-way merge guide must be made with caution because it determines the continuation of the execution (the scenario to follow): this choice could therefore be made by one of the actors involved in the execution of the process (maybe the process owner). This opens an interesting perspective on the verification of the specifications produced with LSAWfP.

The authors of \cite{tchoupeAtemkeng2} have proposed a grammatical model for structured cooperative editing. They also proposed an algorithm for projecting a grammar according to a view. Their algorithm proceeds by successive rewrites of the productions whose right hand side contains invisible symbols: the invisible symbols are replaced by the right hand sides of productions for which they appear on the left hand side. The rewrites are made until there are no more invisible symbols. As in our case, (re)structuring symbols are added if necessary and the initial model must not admit recursion: it therefore seems that the projection of grammatical models admitting recursive symbols is an interesting avenue of research. However, our GMWf projection algorithm is completely based on the artifact projection algorithm. Rewritings of the productions are thus implicitly realized during the projection of the artifacts.

In a more general perspective, Foster et al. \cite{foster2005combinators, foster2007combinators} proposed a solution to the view update problem in the case of tree-structured data. More specifically, they offer a domain-specific programming language in which all expressions denote bi-directional transformations on trees. In a sense, these transformations make it possible to project a so-called "concrete tree" in order to obtain a simplified view (a so-called "abstract view") of it. In the other direction, transformation operations allow to merge a modified abstract view with the original concrete tree to obtain a modified concrete tree. The algorithms proposed by Foster et al. manipulate unranked trees (i.e trees with unranked nodes) while ours only manipulate ranked trees. They are therefore not interested by documents models (grammars) which are an essential tools in our study. Let's note also that, the  concept of views they use in their work is  different from ours; it's  rather more close to the one encountered in works on databases \cite{caruccio2016synchronization, alwehaibi2018rule, meier2019nosql, horn2020language, segoufin2020projection}.

\subsection{Confidentiality using Views in BPM Approaches}
We are not aware of many studies that have looked at the consistency of views in the decentralized execution of processes using BPM technology. We hereby present a few that have mentioned the concept of views.

The studies in \cite{liu2003workflow, chebbi2006view, eshuis2008constructing, zhao2008process, yongchareon2010process, yongchareon2011artifact, yongchareon2015view} propose mechanisms for the construction of views guaranteeing the stability of the base models. Here, the views are in fact, process models that guarantee a certain degree of confidentiality. The studies in \cite{liu2003workflow, chebbi2006view, eshuis2008constructing, zhao2008process} apply in the case of process-centric workflows, while those in \cite{yongchareon2010process, yongchareon2011artifact, yongchareon2015view} apply in the case of artifact-centric workflows. The main difference between these studies and ours is that, we manipulate trees and grammatical models where they are interested in arbitrary graphs or stack automata.

In the SwinDeW \cite{junYan06} approach to decentralised workflow management, the authors propose a "know what you should know" policy to manage confidentiality. According to this policy, a workflow is partitioned (projected) into individual tasks after it is modelled completely (using any workflow language), and definition of individual tasks is then distributed to appropriate peers for storage. Unlike ours, the SwinDeW approach is not artifact-centric and the authors do not really propose projection algorithms; rather, they propose a formalism for modelling each task in order to facilitate their distribution and decentralized execution.

Hull et al. \cite{hull2009facilitating} proposed a new approach to interoperation of organizations hubs based on business artifacts. It provides a centralized point where stakeholders can access data of common interest and check the current status of an aggregate process. They proposed three kinds of access restrictions namely \textit{windows}, \textit{views} and \textit{Create-Read-Update-Delete (CRUD)}. "Windows" provide a mechanism to restrict which artifacts a stakeholder can see; "views" provide a mechanism to restrict what parts of an artifact a stakeholder can see; and the CRUD is used to restrict the ways that stakeholders can read and modify artifacts. This approach differs from ours by the fact that it is centralized and that its confidentiality policy is only interested in the artifacts and not in their models.

\section{Conclusion}
\label{Conclusion}
In this work, we have presented the LSAWfP language for the specification of administrative workflow processes using grammatical models. We then presented a decentralized and artifact-centric execution model (\textit{P2P-WfMS-View}) of the workflow processes specified using LSAWfP. Based on the principles of this model, we proposed versions of its key algorithms (\textit{algorithm for projecting an artifact}, \textit{algorithm for projecting a GMWf} and \textit{algorithm for expanding a partial replica}). The proposed algorithms are perfectly usable since we have proven the stability of our main mathematical tools when using them. We have implemented them in Haskell and tested them with very satisfactory results. 
However, in order for our algorithms to produce the expected results, we have made some assumptions. Notably the non-recursivity of GMWf assumption, which had the direct effect of limiting a little bit the expressiveness of LSAWfP. An interesting perspective of this work therefore consists in proposing other versions of the algorithms presented here, which would take up the same fundamental principles while raising the non-recursivity of GMWf assumption in order to offer more facility to the designers of GMAWfP.

\bibliographystyle{fundam}
\bibliography{bibliography}

\begin{thebibliography}{10}
\providecommand{\url}[1]{\texttt{#1}}
\providecommand{\urlprefix}{URL }
\expandafter\ifx\csname urlstyle\endcsname\relax
  \providecommand{\doi}[1]{doi:\discretionary{}{}{}#1}\else
  \providecommand{\doi}{doi:\discretionary{}{}{}\begingroup
  \urlstyle{rm}\Url}\fi
\providecommand{\eprint}[2][]{\url{#2}}

\bibitem{mcCready}
McCready S.
\newblock \uppercase{t}here is more than one \uppercase{k}ind of
  \uppercase{w}orkflow \uppercase{s}oftware.
\newblock \emph{Computerworld}, 1992.
\newblock \textbf{2}.

\bibitem{van1998application}
Van Der~Aalst WMP.
\newblock \uppercase{t}he \uppercase{a}pplication of \uppercase{p}etri
  \uppercase{n}ets to \uppercase{w}orkflow \uppercase{m}anagement.
\newblock \emph{Journal of Circuits, Systems, and Computers}, 1998.
\newblock \textbf{8}(01):21--66.

\bibitem{van2013business}
Van Der~Aalst WMP.
\newblock Business process management: a comprehensive survey.
\newblock \emph{ISRN Software Engineering}, 2013.
\newblock \textbf{2013}.

\bibitem{dumas2018fundamental}
Dumas M, La~Rosa M, Mendling J, Reijers HA.
\newblock Fundamentals of Business Process Management, Second Edition.
\newblock Springer, 2018.
\newblock ISBN 978-3-662-56508-7.

\bibitem{van2001proclets}
Van Der~Aalst WMP, Barthelmess P, Ellis CA, Wainer J.
\newblock Proclets: A framework for lightweight interacting workflow processes.
\newblock \emph{International Journal of Cooperative Information Systems},
  2001.
\newblock \textbf{10}(04):443--481.

\bibitem{badouel14}
Badouel E, H{\'{e}}lou{\"{e}}t L, Kouamou GE, Morvan C.
\newblock \uppercase{a} \uppercase{g}rammatical \uppercase{a}pproach to
  \uppercase{d}ata-centric \uppercase{c}ase \uppercase{m}anagement in a
  \uppercase{d}istributed \uppercase{c}ollaborative \uppercase{e}nvironment.
\newblock \emph{CoRR}, 2014.
\newblock \textbf{abs/1405.3223}.
\newblock \eprint{1405.3223}, \urlprefix\url{http://arxiv.org/abs/1405.3223}.

\bibitem{zekeng2020alanguage}
Zekeng~Ndadji MM, Tchoup{\'e}~Tchendji M, Tayou~Djamegni C, Parigot D.
\newblock A Language for the Specification of Administrative Workflow Processes
  with Emphasis on Actors' Views.
\newblock In: Gervasi O. et al. (eds) Computational Science and Its
  Applications - ICCSA 2020. ICCSA 2020. Lecture Notes in Computer Science,
  volume 12254. Springer, 2020 pp. 231--245.

\bibitem{zekeng2020lsawfp}
Zekeng~Ndadji MM, Tchoup{\'e}~Tchendji M, Tayou~Djamegni C, Parigot D.
\newblock A Grammatical Model for the Specification of Administrative Workflow
  using Scenario as Modelling Unit.
\newblock In: Florez H., Misra S. (eds) Applied Informatics. ICAI 2020.
  Communications in Computer and Information Science, volume 1277. Springer,
  2020 pp. 131--145.

\bibitem{ndadji2020language}
Zekeng~Ndadji MM, Tchoup{\'e}~Tchendji M, Tayou~Djamegni C, Parigot D.
\newblock A Language and Methodology based on Scenarios, Grammars and Views,
  for Administrative Business Processes Modelling.
\newblock \emph{ParadigmPlus}, 2020.
\newblock \textbf{1}(3):1--22.

\bibitem{esparza2016reduction}
Esparza J, Hoffmann P.
\newblock Reduction rules for colored workflow nets.
\newblock In: International Conference on Fundamental Approaches to Software
  Engineering. Springer, 2016 pp. 342--358.

\bibitem{cartledge2020system}
Cartledge D, Danis H, Jacobus J, Burke S, Millington M, Sierra S, Smith G.
\newblock System and method for workflow management, 2020.
\newblock US Patent 10,546,272.

\bibitem{saito2005optimistic}
Saito Y, Shapiro M.
\newblock Optimistic replication.
\newblock \emph{ACM Computing Surveys (CSUR)}, 2005.
\newblock \textbf{37}(1):42--81.

\bibitem{shapiro2018Optimistic}
Shapiro M.
\newblock Optimistic Replication and Resolution.
\newblock In: Encyclopedia of Database Systems, Second Edition.
  Springer-Verlag, 2018.
\newblock \doi{10.1007/978-1-4614-8265-9\_258}.

\bibitem{badouelTchoupeCmcs}
Badouel E, Tchoup{\'{e}}~Tchendji M.
\newblock \uppercase{m}erging \uppercase{h}ierarchically-\uppercase{s}tructured
  \uppercase{d}ocuments in \uppercase{w}orkflow \uppercase{s}ystems.
\newblock \emph{Electronic Notes in Theoretical Computer Science}, 2008.
\newblock \textbf{203}(5):3--24.

\bibitem{theseTchoupe}
Tchoup\'{e}~Tchendji M.
\newblock \uppercase{u}ne \uppercase{a}pproche \uppercase{g}rammaticale pour la
  \uppercase{f}usion des \uppercase{r}{\'e}plicats \uppercase{p}artiels d'un
  \uppercase{d}ocument \uppercase{s}tructur{\'e}: \uppercase{a}pplication {\`a}
  l'\uppercase{{\'e}}dition \uppercase{c}oop{\'e}rative \uppercase{a}synchrone.
\newblock Phd thesis, Universit{\'e} de Rennes I (France), Universit{\'e} de
  Yaound{\'e} I (Cameroun), 2009.

\bibitem{tchoupeAtemkeng2}
Tchoup{\'e}~Tchendji M, Djeumen~D R, Atemkeng~T M.
\newblock \uppercase{a} \uppercase{s}table and \uppercase{c}onsistent
  \uppercase{d}ocument \uppercase{m}odel \uppercase{s}uitable for
  \uppercase{a}synchronous \uppercase{c}ooperative \uppercase{e}dition.
\newblock \emph{Journal of Computer and Communications}, 2017.
\newblock \textbf{5}(08):69.

\bibitem{tchoupezekeng2016reconciliation}
Tchoup{\'e}~Tchendji M, Zekeng~Ndadji MM.
\newblock R{\'e}conciliation par consensus des mises {\`a} jour des
  r{\'e}pliques partielles d'un document structur{\'e}.
\newblock In: CARI 2016 Proceedings, volume~1. 2016 pp. 84--96.

\bibitem{tchoupeZekeng2017}
Tchoup{\'e}~Tchendji M, Zekeng~Ndadji MM.
\newblock \uppercase{t}ree \uppercase{a}utomata for \uppercase{e}xtracting
  \uppercase{c}onsensus from \uppercase{p}artial \uppercase{r}eplicas of a
  \uppercase{s}tructured \uppercase{d}ocument.
\newblock \emph{Journal of Software Engineering and Applications}, 2017.
\newblock \textbf{10}(05):432--456.

\bibitem{zekengTchoupe2018}
Zekeng~Ndadji MM, Tchoup{\'e}~Tchendji M.
\newblock \uppercase{a} \uppercase{s}oftware \uppercase{a}rchitecture for
  \uppercase{c}entralized \uppercase{m}anagement of \uppercase{s}tructured
  \uppercase{d}ocuments in a \uppercase{c}ooperative \uppercase{e}diting
  \uppercase{w}orkflow.
\newblock In: Innovation and Interdisciplinary Solutions for Underserved Areas,
  pp. 279--291. Springer, 2018.

\bibitem{BPMN}
Model BP.
\newblock \uppercase{n}otation \uppercase{(BPMN)} version 2.0.
\newblock \emph{OMG Specification, Object Management Group}, 2011.
\newblock pp. 22--31.

\bibitem{van2005yawl}
Van Der~Aalst WMP, Ter~Hofstede AHM.
\newblock YAWL: yet another workflow language.
\newblock \emph{Information systems}, 2005.
\newblock \textbf{30}(4):245--275.

\bibitem{van2015business}
Van Der~Aalst WMP.
\newblock Business process management as the "Killer Ap" for Petri nets.
\newblock \emph{Software \& Systems Modeling}, 2015.
\newblock \textbf{14}(2):685--691.

\bibitem{zur2013much}
Zur~Muehlen M, Recker J.
\newblock How much language is enough? Theoretical and practical use of the
  business process modeling notation.
\newblock In: Seminal Contributions to Information Systems Engineering, pp.
  429--443. Springer, 2013.

\bibitem{borger2012approaches}
B{\"o}rger E.
\newblock Approaches to modeling business processes: a critical analysis of
  BPMN, workflow patterns and YAWL.
\newblock \emph{Software \& Systems Modeling}, 2012.
\newblock \textbf{11}(3):305--318.

\bibitem{divitini2001inter}
Divitini M, Hanachi C, Sibertin-Blanc C.
\newblock Inter-Organizational Workflows for Enterprise Coordination.
\newblock In: Coordination of Internet agents, pp. 369--398. Springer, 2001.

\bibitem{nigam2003business}
Nigam A, Caswell NS.
\newblock Business artifacts: An approach to operational specification.
\newblock \emph{IBM Systems Journal}, 2003.
\newblock \textbf{42}(3):428--445.

\bibitem{abi2016towards}
Assaf MA.
\newblock Towards an integration system for artifact-centric processes.
\newblock In: Proceedings of the 2016 on SIGMOD'16 PhD Symposium. ACM, 2016 pp.
  2--6.

\bibitem{deutsch2014automatic}
Deutsch A, Hull R, Vianu V.
\newblock Automatic verification of database-centric systems.
\newblock \emph{ACM SIGMOD Record}, 2014.
\newblock \textbf{43}(3):5--17.

\bibitem{hull2009facilitating}
Hull R, Narendra NC, Nigam A.
\newblock Facilitating workflow interoperation using artifact-centric hubs.
\newblock In: Service-Oriented Computing, pp. 1--18. Springer, 2009.

\bibitem{lohmann2010artifact}
Lohmann N, Wolf K.
\newblock Artifact-centric choreographies.
\newblock In: International Conference on Service-Oriented Computing. Springer,
  2010 pp. 32--46.

\bibitem{assaf2017continuous}
Assaf MA, Badr Y, Amghar Y.
\newblock A Continuous Query Language for Stream-Based Artifacts.
\newblock In: International Conference on Database and Expert Systems
  Applications. Springer, 2017 pp. 80--89.

\bibitem{assaf2018generating}
Assaf MA, Badr Y, El~Khoury H, Barbar K.
\newblock Generating Database Schemas from Business Artifact Models.
\newblock \emph{I.J. Information Technology and Computer Science}, 2018.
\newblock \textbf{2}:10--17.
\newblock \doi{10.5815/ijitcs.2018.02.02}.

\bibitem{boaz2013bizartifact}
Boaz D, Limonad L, Gupta M.
\newblock BizArtifact: Artifact-centric Business Process Management, June 2013,
  2013.
\newblock \urlprefix\url{https://sourceforge.net/projects/bizartifact/,
  accessed 12 December 2019}.

\bibitem{badouel2015active}
Badouel E, H{\'e}lou{\"e}t L, Kouamou GE, Morvan C, Fondze~Jr NR.
\newblock Active workspaces: distributed collaborative systems based on guarded
  attribute grammars.
\newblock \emph{ACM SIGAPP Applied Computing Review}, 2015.
\newblock \textbf{15}(3):6--34.

\bibitem{tomMens}
Mens T.
\newblock \uppercase{a} \uppercase{s}tate-of-the-\uppercase{a}rt
  \uppercase{s}urvey on \uppercase{s}oftware \uppercase{m}erging.
\newblock \emph{Journal of {IEEE} Transactions on Software Engineering}, 2002.
\newblock \textbf{28}(5):449--462.

\bibitem{foster2005combinators}
Foster JN, Greenwald MB, Moore JT, Pierce BC, Schmitt A.
\newblock Combinators for bi-directional tree transformations: a linguistic
  approach to the view update problem.
\newblock \emph{ACM SIGPLAN Notices}, 2005.
\newblock \textbf{40}(1):233--246.

\bibitem{foster2007combinators}
Foster JN, Greenwald MB, Moore JT, Pierce BC, Schmitt A.
\newblock Combinators for bidirectional tree transformations: A linguistic
  approach to the view-update problem.
\newblock \emph{ACM Transactions on Programming Languages and Systems
  (TOPLAS)}, 2007.
\newblock \textbf{29}(3):17--es.

\bibitem{caruccio2016synchronization}
Caruccio L, Polese G, Tortora G.
\newblock Synchronization of queries and views upon schema evolutions: A
  survey.
\newblock \emph{ACM Transactions on Database Systems (TODS)}, 2016.
\newblock \textbf{41}(2):1--41.

\bibitem{alwehaibi2018rule}
Alwehaibi A, Atay M.
\newblock A rule-based relational xml access control model in the presence of
  authorization conflicts.
\newblock In: Information Technology-New Generations, pp. 311--319. Springer,
  2018.

\bibitem{meier2019nosql}
Meier A, Kaufmann M.
\newblock Nosql databases.
\newblock In: SQL \& NoSQL Databases, pp. 201--218. Springer, 2019.

\bibitem{horn2020language}
Horn R, Fowler S, Cheney J.
\newblock Language-Integrated Updatable Views (Extended version).
\newblock \emph{arXiv preprint arXiv:2003.02191}, 2020.

\bibitem{segoufin2020projection}
Segoufin L, Vianu V.
\newblock Projection Views of Register Automata.
\newblock In: Proceedings of the 39th ACM SIGMOD-SIGACT-SIGAI Symposium on
  Principles of Database Systems. 2020 pp. 299--313.

\bibitem{liu2003workflow}
Liu DR, Shen M.
\newblock Workflow modeling for virtual processes: an order-preserving
  process-view approach.
\newblock \emph{Information Systems}, 2003.
\newblock \textbf{28}(6):505--532.

\bibitem{chebbi2006view}
Chebbi I, Dustdar S, Tata S.
\newblock The view-based approach to dynamic inter-organizational workflow
  cooperation.
\newblock \emph{Data \& Knowledge Engineering}, 2006.
\newblock \textbf{56}(2):139--173.

\bibitem{eshuis2008constructing}
Eshuis R, Grefen P.
\newblock Constructing customized process views.
\newblock \emph{Data \& Knowledge Engineering}, 2008.
\newblock \textbf{64}(2):419--438.

\bibitem{zhao2008process}
Zhao X, Liu C, Sadiq W, Kowalkiewicz M.
\newblock Process view derivation and composition in a dynamic collaboration
  environment.
\newblock In: OTM Confederated International Conferences" On the Move to
  Meaningful Internet Systems". Springer, 2008 pp. 82--99.

\bibitem{yongchareon2010process}
Yongchareon S, Liu C.
\newblock A process view framework for artifact-centric business processes.
\newblock In: OTM Confederated International Conferences" On the Move to
  Meaningful Internet Systems". Springer, 2010 pp. 26--43.

\bibitem{yongchareon2011artifact}
Yongchareon S, Liu C, Zhao X.
\newblock An artifact-centric view-based approach to modeling
  inter-organizational business processes.
\newblock In: International Conference on Web Information Systems Engineering.
  Springer, 2011 pp. 273--281.

\bibitem{yongchareon2015view}
Yongchareon S, Yu J, Zhao X, et~al.
\newblock A view framework for modeling and change validation of
  artifact-centric inter-organizational business processes.
\newblock \emph{Information systems}, 2015.
\newblock \textbf{47}:51--81.

\bibitem{junYan06}
Yan J, Yang Y, Raikundalia GK.
\newblock \uppercase{s}win\uppercase{d}e\uppercase{w}-a \uppercase{p2p-b}ased
  \uppercase{d}ecentralized \uppercase{w}orkflow \uppercase{m}anagement
  \uppercase{s}ystem.
\newblock \emph{{IEEE} Trans. Systems, Man, and Cybernetics, Part {A}}, 2006.
\newblock \textbf{36}(5):922--935.

\end{thebibliography}


\end{document}